\documentclass[10pt, journal, letterpaper]{IEEEtran}

\usepackage[a-1b]{pdfx}

\usepackage{listings}
\usepackage{amsmath, amssymb, amsthm, amsfonts}
\usepackage{mathtools}
\usepackage{xcolor}

\usepackage{hyperref}
\hypersetup{colorlinks=true,
            linkcolor=black,
            citecolor=black,
            urlcolor=black}

\usepackage[capitalize]{cleveref}

\usepackage{enumerate}
\usepackage{cite}
\usepackage{url}
\usepackage{graphicx}
\graphicspath{ {Figures/} }
\usepackage[font=normalsize]{subfig}
\usepackage{tikz}
\usepackage{titlesec}
\titlespacing*{\section}{0pt}{1.5ex plus 1ex minus 0.2ex}{0.5ex plus 0.2ex}
\titlespacing*{\subsection}{0pt}{1.5ex plus 1ex minus 0.2ex}{0.5ex plus 0.2ex}
\newtheoremstyle{mytheorem}
{3pt}                
{3pt}                
{}        
{}                
{\bfseries}       
{.}               
{ }               
{}                
\theoremstyle{mytheorem}
\newtheorem{thm}{Theorem}
\newtheorem{prop}{Proposition}
\newtheorem{lem}{Lemma}

\theoremstyle{definition}
\newtheorem{define}{Definition}
\newtheorem{assume}{Assumption}
\theoremstyle{remark}
\newtheorem{rem}{Remark}

\def\ba{\begin{array}}
\def\ea{\end{array}}

\def\bb{\mathbb}
\def\mc{\mathcal}
\def\T{\intercal}
\def\D{\text{D}}
\def\ol{\overline}
\def\ul{\underline}

\DeclareMathOperator*{\argmin}{arg\,min}

\newcommand\eqdef{\coloneqq}

\title{Eco-driving Incentive Mechanisms for Mitigating Emissions in Urban Transportation}
\author{
M.~Umar~B.~Niazi,~\IEEEmembership{Member,~IEEE,}
Jung-Hoon~Cho,~\IEEEmembership{Graduate Student Member,~IEEE,}
Munther~A.~Dahleh,~\IEEEmembership{Fellow,~IEEE,}
Roy~Dong,~\IEEEmembership{Member,~IEEE,}
and~Cathy~Wu,~\IEEEmembership{Member,~IEEE,}
\thanks{M. U. B. Niazi is with the Laboratory for Information and Decision Systems, Department of Electrical Engineering and Computer Science, Massachusetts Institute of Technology, Cambridge, MA 02139, USA, and with the Division of Decision and Control Systems, Digital Futures, KTH Royal Institute of Technology, SE-100 44 Stockholm, Sweden. Emails: \texttt{niazi@mit.edu}, \texttt{mubniazi@kth.se}}
\thanks{J.-H. Cho and C. Wu are with the Laboratory for Information and Decision Systems, Department of Civil and Environmental Engineering, Massachusetts Institute of Technology, Cambridge, MA 02139, USA. Email: \texttt{jhooncho@mit.edu}, \texttt{cathywu@mit.edu}}
\thanks{M. A. Dahleh is with the Laboratory for Information and Decision Systems, Department of Electrical Engineering and Computer Science, Massachusetts Institute of Technology, Cambridge, MA 02139, USA. Email: \texttt{dahleh@mit.edu}}
\thanks{R. Dong is with the Department of Industrial and Enterprise Systems Engineering, University of Illinois at Urbana-Champaign, Urbana, IL 61801, USA. Email: \texttt{roydong@illinois.edu}}
\thanks{This work was partially supported by the European Union's Horizon Research and Innovation Programme under Marie Sk\l{}odowska-Curie grant agreement No. 101062523, the National Science Foundation (NSF) under grant number 2149548, and the Kwanjeong scholarship.}
}

\begin{document}

\maketitle

\begin{abstract}
    This paper develops incentive mechanisms for promoting eco-driving with the overarching goal of minimizing emissions in transportation networks. 
    The system operator provides drivers with energy-efficient driving guidance throughout their trips and measures compliance through vehicle telematics that capture how closely drivers follow this guidance. 
    Drivers optimize their behaviors based on personal trade-offs between travel times and emissions. To design effective incentives, the operator elicits driver preferences regarding trip urgency and willingness to eco-drive, while determining optimal budget allocations and eco-driving recommendations.
    Two distinct settings based on driver behavior are analyzed. 
    When drivers report their preferences truthfully, an incentive mechanism ensuring obedience (drivers find it optimal to follow recommendations) is designed by implementing eco-driving recommendations as a Nash equilibrium. 
    When drivers may report strategically, the mechanism is extended to be both obedient and truthful (drivers find it optimal to report truthfully). 
    Unlike existing works that focus on congestion or routing decisions in transportation networks, our framework explicitly targets emissions reduction by incentivizing drivers. 
    The proposed mechanism addresses both strategic behavior and network effects arising from driver interactions, without requiring the operator to reveal system parameters to the drivers. 
    Numerical simulations demonstrate the effects of budget constraints, driver types, and strategic misreporting on equilibrium outcomes and emissions reduction.
\end{abstract}

\begin{IEEEkeywords}
    Incentive design, eco-driving, urban transportation, Nash equilibrium, obedience, truthfulness.
\end{IEEEkeywords}

\section{Introduction}

The transportation sector is a major contributor to climate change, accounting for a significant portion — between a quarter and a third \cite{epa-ghg2025} — of global greenhouse gas emissions. 
Urban transportation, in particular, poses a significant and growing threat, substantially contributing to these emissions with projections indicating an increase of 16-50\% by 2050 \cite{ipcc-2022}, even with improvements in vehicle technology and fuel economy.
This necessitates immediate action regarding financial investments and the adoption of readily available solutions for sustainable urban transportation \cite{liotta2023}. 

Numerous strategies have been employed to improve fuel efficiency and decrease emissions in on-road vehicles. 
While advancements in engine technology, electrification, and autonomous vehicles hold promise for sustainable transportation, these solutions require significant time and investment. 
Presently, even electric vehicles rely on an energy grid and raw materials that may not be entirely emissions-free and/or sustainable. Under these circumstances, eco-driving stands out as an immediate, cost-effective, and highly efficient means of reducing emissions from urban transportation \cite{huang2018}.

Eco-driving comprises techniques that optimize vehicle operation and driver behavior to improve fuel/energy efficiency, which is directly proportional to emissions \cite{jayawardana2022}. 
These techniques include smooth acceleration and deceleration, avoiding unnecessary idling, and maintaining steady speeds.
Eco-driving can substantially reduce emissions of internal combustion engine (ICE) vehicles, ranging from 10\% to as high as 45\% \cite{sivak2012}. 
For electric vehicles (EVs), eco-driving translates to energy-efficient driving \cite{zhang2015}. 
By reducing the overall energy consumption of EVs through eco-driving, the demand on the power grid decreases, which can lead to a reduction in indirect emissions depending on the grid's emissions factor \cite{holdway2010}.

While eco-driving offers significant environmental benefits, it can sometimes lead to longer travel times \cite{niazi2024, cho2024}.
This can be a barrier for some drivers to adopt eco-driving because of prioritizing shorter travel times over lower emissions. 
Therefore, transportation system operators (TSOs) need to consider implementing incentive mechanisms that reward drivers for adopting eco-driving practices. 
These eco-driving incentive mechanisms (EDIMs) not only encourage drivers to choose eco-friendly driving styles but also contribute to achieving the overall emission reduction goals in transportation networks.

A considerable amount of evidence supports the effectiveness of incentive mechanisms in promoting eco-driving. 
For instance, \cite{lai2015} observed a reduction of over 10\% in fuel consumption and emissions after monetary incentives were introduced to bus drivers for eco-driving.
Comparable results were obtained by \cite{liimatainen2011} and \cite{schall2017}, where logistic companies incentivized heavy-duty vehicle drivers. 
Behavioral studies, such as \cite{mcconky2018} and \cite{vaezipour2019}, demonstrate that incentives are more effective in changing driver behavior than merely providing informational indicators for eco-driving through in-vehicle interfaces.
However, the incentive mechanisms presented in this body of literature are overly simplistic and do not cater to various driver types with differing preferences.

To address the variability in driver preferences, our previous work \cite{niazi2024} developed a traffic simulation-assisted EDIM to minimize emissions in a transportation network. 
After gathering private information from the drivers about their trips and preferences, the TSO computes feasible eco-driving strategies for each driver that minimize the overall emissions of the network subject to budget and obedience constraints. 
Such an EDIM allows the TSO to provide personalized incentives and eco-driving recommendations to the drivers. 
However, \cite{niazi2024} does not account for the interactions between the drivers when they eco-drive in a transportation network. 
In reality, the driving policy of each driver impacts the traffic around her\footnote{We will refer to a driver by she/her and the TSO by he/him.}, influencing the driving policies of other drivers, and vice versa. 
Such an interaction is natural because the drivers usually aim to conform with the surrounding traffic \cite{schwarting2019, ozkan2021, jayawardana2023}. 
Moreover, \cite{niazi2024} assumes that drivers provide accurate information about their preferences, which may not hold when they are strategic.

In the present paper, we design EDIMs for promoting eco-driving by considering two key challenges: i)~interactions between the drivers who share routes, and ii)~drivers strategically reporting their preferences to the TSO to maximize their incentives.
The incentive mechanism induces an eco-driving game where drivers choose eco-driving levels that minimize a combination of their travel times and emissions, and maximize their respective incentives.
Assuming that drivers are truthful or the TSO knows their preferences/types, we design the so-called \textit{first-best} EDIM that minimizes the overall emissions by incentivizing drivers under budget constraints. 
The first-best EDIM implements the recommended eco-driving profile as a Nash equilibrium, which is shown to ensure obedience, i.e., the drivers find it optimal to adhere to the recommended eco-driving levels provided by the TSO. 
When drivers may strategically report their types, we design the so-called \textit{second-best} EDIM that, in addition to obedience, also ensures truthfulness, i.e., the drivers find it optimal to report their types truthfully.

The paper is organized as follows.
\cref{sec:related-work} reviews the related literature.
\cref{sec:model} describes our model of an eco-driving incentive mechanism. \cref{sec:mech-known-types} and \cref{sec:mech-unknown-types} propose our first-best and second-best EDIMs. \cref{sec:simulation} presents numerical simulations showing the effect of misreporting and the total incentive budget on eco-driving and emissions. Finally, \cref{sec:conclusion} presents the concluding remarks.

\textit{Notations.} 
Non-negative and positive real numbers are denoted as $\bb R_+$ and $\bb R_+^*$. 
Let $v=(v_1,\dots,v_n)\eqdef (v_i,v_{-i})$, where $v_i$ is the $i$-th element of $v$ and $v_{-i}\eqdef(v_1,\dots,v_{i-1},v_{i+1},\dots,v_n)$. 
For $\mc I\subseteq\{1,\dots,n\}$, $v_{\mc I}=(v_j)_{j\in\mc I}$. The interior of a closed set $S\subset\bb R^n$ is denoted by $\text{int}(S)$. 
Given sets $S_1,\dots,S_n$, denote $S_{-i} \eqdef \prod_{j\neq i} S_j$. 
Given a differentiable function $f:\bb R^n\to \bb R$, we denote its partial derivative with respect to $v_i$ as $\D_{v_i} f(v)$ and its gradient $\nabla_v f(v) = (\D_{v_1} f(v),\dots,\D_{v_n} f(v))$, where $v\in\bb R^n$. Finally, for $a\in\bb R$, we denote $|a|_+ \eqdef \max(0,a)$.

\section{Related Work} \label{sec:related-work}
The literature on incentive mechanisms for transportation and other infrastructure systems presents various approaches to influence agent behavior through different information structures and design constraints.
These approaches fundamentally differ in two key aspects: strategic versus non-strategic agents (e.g., drivers) and the directionality of information flow between system operators and agents.

In transportation, a significant body of work examines settings where TSOs use state information as a mechanism to influence traffic patterns. 
In this regard, \cite{wu2019} studies optimal information design from a Bayesian persuasion perspective, where operators send noisy signals about uncertain network states to regulate traffic flows. 
This framework is extended by \cite{zhu2022}, which develops computational approaches for both public and private signaling policies in nonatomic routing games. 
Similarly, \cite{massicot2021} establishes hierarchies between different information provision strategies, comparing public signaling with private recommendations. 
Along similar lines, \cite{ferguson2024} explores how information signaling, when combined with monetary incentives, affects system performance in Bayesian congestion games. 
These works share a common feature: the system operator leverages statistical knowledge about driver populations to optimize traffic outcomes through information disclosure, without requiring drivers to report their private preferences.

In contrast, mechanisms involving bidirectional information exchange introduce additional design challenges. 
For instance, \cite{bonatti2023, bonatti2024} examine the joint problem of information and mechanism design when selling information to competing agents with private types, establishing structural properties of optimal mechanisms while ensuring both obedience and truthfulness. 
Building on similar principles, \cite{satchidanandan2022-2, satchidanandan2023} develop two-stage mechanisms, where agents strategically report probabilistic descriptions of their future parameters in the first stage and realized values in the second stage.
Extension of this multi-stage framework to indirect mechanisms and non-myopic agents \cite{dahleh2024} demonstrates that inducing truthfulness imposes strict structural constraints on incentive functions.
These multi-stage mechanisms focus on ensuring truthful behavior from the agents without requiring them to comply with recommended actions.
Similarly, VCG-type mechanisms, e.g., \cite{satchidanandan2022}, focus purely on preference elicitation without providing action recommendations. 

Our work occupies a unique position in this domain. 
Unlike the information design literature, where TSOs share traffic state information with non-strategic drivers, we focus on eliciting private preferences from strategic drivers without revealing the system state. 
The system operator provides eco-driving recommendations and incentives to the drivers after gathering their private data. 
Our mechanisms are obedient and, in the strategic setting, truthful. 
This creates a joint design problem where recommendation and incentive functions must be optimized under dual constraints, while explicitly targeting emissions reduction rather than traditional objectives like flow maximization. 
This combination of strategic behavior, limited information sharing, and emissions-focused objectives distinguishes our work from existing approaches.

\section{Model of Eco-Driving Incentive Mechanism}
\label{sec:model}

This section explores the decision-making process that drivers utilize for selecting their eco-driving policies and the incentive mechanism model. 

\subsection{Eco-driving Policy and Interactions} 
Consider a set of drivers $\mc N=\{1,\dots,n\}$ in a transportation network $\mc G=(\mc I, \mc L)$, where $\mc I$ is the set of intersections and $\mc L\subseteq \mc I\times \mc I$ is the set of links/roads. Each driver~$i$ chooses a route $\mc R_i=\{o_i,\dots,d_i\}$ to drive between her origin $o_i\in \mc L$ and her destination $d_i\in \mc  L$, where $d_i\neq o_i$. 
Driver~$j$ influences $i$ if (i)~they share links on their routes, i.e., $\mc R_i\cap \mc R_j \neq \emptyset$ and there is a link $l\in\mc R_i\cap \mc R_j$ that $i$ and $j$ occupy at the same time, and (ii)~$j$ is positioned ahead of $i$ on link~$l$.

The driving policy adopted by driver~$i$, for $i\in\mc N$, on her route~$\mc R_i$ is mapped to an eco-driving level~$a_i\in A_i = [0,1]$. 
Knowing the vehicle types, origins, destinations, and routes of drivers, the TSO uses a digital platform to provide drivers with personalized eco-driving guidance that results in minimal emissions on their journeys. If driver~$i$ fully complies with this eco-driving guidance, then we say that her eco-driving level is $a_i=1$. Deviations from this eco-driving guidance result in $a_i<1$. A driver may choose $a_i<1$ to, for instance, reduce her travel time. 
Thus, drivers adopt driving policies to find optimal trade-offs between their emissions and travel time.

\begin{rem}
    A wealth of research exists on designing controllers for fuel-efficient or eco-friendly driving \cite{kamal2010, padilla2018, han2019}. These studies span signalized intersections \cite{mintsis2020, sun2020, jayawardana2022}, extend to connected vehicles \cite{turri2016, homchaudhuri2016}, and even cover urban transportation networks \cite{de-nunzio2016}, all considering varying traffic conditions. This research primarily focuses on eco-driving for autonomous vehicles, which can perfectly comply with the eco-driving policies computed by onboard controllers. However, human drivers may not comply with the eco-driving guidance because of the misalignment of their objectives with the TSO. The aim of the TSO is to minimize the overall emissions of the network. In addition to minimizing emissions or fuel consumption, drivers' objectives may also include minimizing their travel time. This misalignment of objectives motivates the need for incentives that encourage drivers to comply with the eco-driving guidance.
    \hfill $\diamond$
\end{rem}

\begin{figure}[!]
    \centering
    \begin{tikzpicture}
        \node at (0,0) {\includegraphics[width=0.95\linewidth]{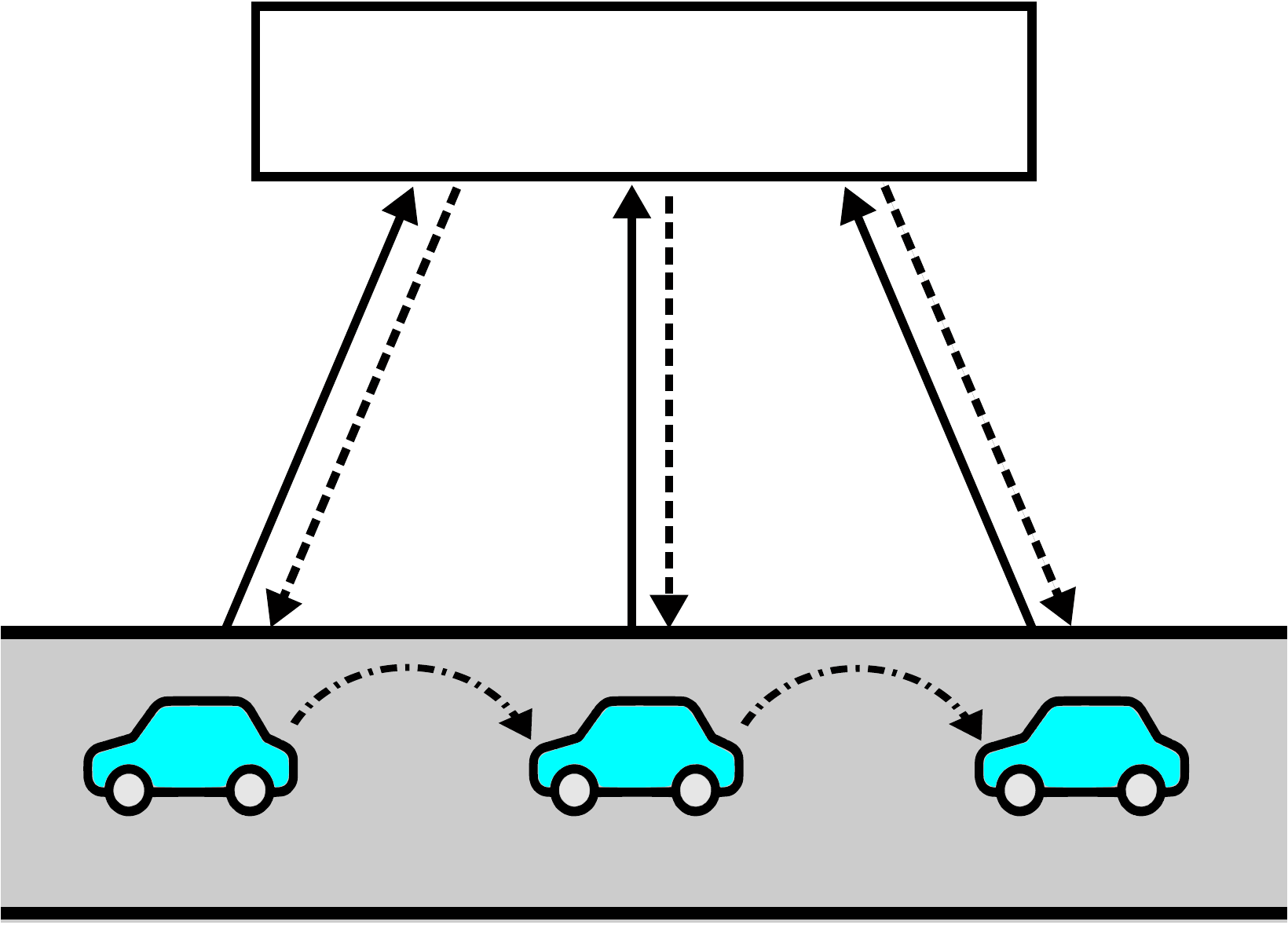}};
        \node[align=center] at (0,2.3) {Transportation System Operator \\ (TSO)};
        \node[rotate=67] at (-2.5,0.3) {\small Reported type $\hat\theta_{i+1}$};
        \node[rotate=-67, text width=3cm, align=center] at (3,0.5) {\small Recommended eco-driving $f_{i-1}$ and incentive $u_{i-1}$};
        \node at (0,-1.85) {\small $i$};
        \node at (-2.9,-1.85) {\small $i+1$};
        \node at (2.9,-1.85) {\small $i-1$};

        \node at (-2.9,-2.6) {\small Eco-driving $a_{i+1}$};
        \node at (0,-2.6) {\small $a_{i}(a_{i+1})$};
        \node at (2.9,-2.6) {\small $a_{i-1}(a_i)$};
    \end{tikzpicture}
    \caption{Interactions between the TSO and the drivers in an EDIM, and among the drivers when they eco-drive.}
    \label{fig:mech}
\end{figure}
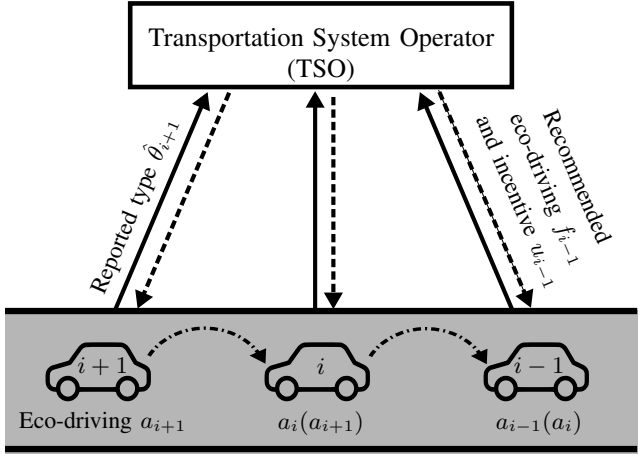

\subsection{Emissions and Travel Time}

For our theoretical analysis of EDIMs, we conceptualize a monitoring phase comprising multiple trips of the drivers in the past where the TSO used vehicle telematics to collect data related to the driving behavior of each driver~$i\in\mc N$.
This abstraction allows us to develop a mathematical framework where the TSO already has estimates of interactions between drivers based on their daily routes and how different eco-driving policies affect emissions and travel times.
In principle, such estimation is feasible in urban settings with periodic traffic patterns during weekdays, where drivers follow regular routes between homes and workplaces during predictable hours.
While perfect estimation would be challenging in practice, our theoretical framework assumes that the TSO has sufficient historical data to estimate the functions of emissions $x_i: A\to X_i$ and travel time $y_i: A\to Y_i$ for each driver~$i\in \mc N$, where $A=\prod_{j=1}^n A_j$ and $X_i,Y_i\subset \bb R_+^*$.

\begin{rem}
    Microscopic models like CMEM \cite{scora2006} or VT-CPFM \cite{park2013} provide established frameworks for mapping driving behavior and traffic conditions to emissions and travel time. 
    These could conceptually be adapted to our eco-driving framework by parameterizing driving behavior as eco-driving level $a_i$ and to some degree representing traffic conditions through $a_{-i}$.
    However, these models require detailed inputs that may be challenging to obtain in practice, and their adaptation would introduce estimation uncertainties that would need to be quantified. 
    Nonetheless, advances in telematics technology can provide the data required for accurate estimation of these functions.
    For instance, \cite{singh2021} develops a deep learning method using telematics data to precisely estimate vehicle emissions within a transportation network. Similarly, \cite{allstrom2012} focuses on macroscopic travel time predictions, whereas \cite{li2017} uses GPS data of vehicles to predict microscopic travel times with high accuracy.
    In the absence of traditional telematics, smartphone-based telematics \cite{wahlstrom2017} offers a viable solution for estimating both emissions and travel times.
    The reader is referred to \cite{gao2022} for a survey on learning from vehicle telematics data.
    \hfill $\diamond$
\end{rem}

Rather than focusing on the specific methods for obtaining emissions and travel time functions, our primary contribution lies in analyzing the structural properties of eco-driving incentive mechanisms while assuming these functions are known.
With this context in place, we make the following assumptions on the functional class of emissions and travel time.

\begin{assume}
    \label{assume:convex-diff}
    For every $i\in\mc N$ and $a_{-i}\in A_{-i}$, the emissions 
    $x_i:A \to X_i$
    and the travel time 
    $y_i:A \to Y_i$
    are convex and differentiable functions with respect to their $i$-th argument $a_i\in A_i$. 
\end{assume}

From a practical perspective, the convexity of $x_i$ reasonably approximates the diminishing returns in emission reduction (see \cref{assume:emission-decreasing}) as eco-driving levels increase, i.e., each incremental improvement in eco-driving behavior tends to yield progressively smaller emission benefits. 
On the other hand, the convexity of $y_i$ means that as the eco-driving level $a_i$ increases from 0 to 1, the incremental travel time penalties become progressively larger.

Let $\mc N_i$ denote the set of drivers that influence $i$.
This set can be determined from the past vehicle trajectory data of drivers.
Then, for every $i\in\mc N$, the emissions $x_i$ and the travel time $y_i$ are functions of the driving policy of $i$ and the driving policies of $\mc N_i$. In reality, the emissions $x_i(a_i,a_{-i}) \eqdef x_i(a_i,a_{\mc N_i})$ and the travel time $y_i(a_i,a_{-i}) \eqdef y_i(a_i,a_{\mc N_i})$, where $a_{\mc N_i}=(a_j)_{j\in \mc N_i}$. However, for simplicity and to avoid clutter, we keep the notation $x_i(a_i,a_{-i})$ and $y_i(a_i,a_{-i})$, which is, for good reasons, prevalent in game theory literature.

\begin{assume}
    \label{assume:emission-decreasing}
    For every $a_{-i}\in A_{-i}$, $x_i(a_i,a_{-i})$ is non-increasing with respect to $a_i\in A_i$.
\end{assume}

\cref{assume:emission-decreasing} considers full compliance with the eco-driving guidance provided by the TSO, i.e., eco-driving level $a_i=1$, as the most energy-efficient policy that results in minimal emissions. That is, given $a_{-i}\in A_{-i}$, $x_i(1,a_{-i})\leq x_i(a_i,a_{-i})$ for every $a_i\in A_i=[0,1]$. In fact, for every $a_i, \hat a_i\in A_i$ with $\hat a_i \geq a_i$, we have $x_i(\hat a_i,a_{-i})\leq x_i(a_i, a_{-i})$. Moreover, due to \cref{assume:convex-diff}, we have $\D_{a_i} x_i(\hat a_i,a_{-i}) \geq \D_{a_i} x_i(a_i,a_{-i})$ for $\hat a_i\geq a_i$.

\subsection{Decision Model of Drivers}
\label{subsec:decision-model}
In the absence of incentives under nominal conditions, we assume that each driver~$i$ chooses an eco-driving level that minimizes her nominal cost $c_i:A\times \Theta_i\to \bb R$ given by
\begin{equation}
    \label{eq:cost-i-nom}
    c_i(a_i,a_{-i},\theta_i) = \theta_i x_i(a_i,a_{-i}) + (1-\theta_i) y_i(a_i,a_{-i})
\end{equation}
where $\theta_i\in \Theta_i=[0,1]$ is the \textit{type} of driver~$i$ that determines the relative importance she places on her emissions over her travel time. Let $\Theta=\prod_{j=1}^n \Theta_j$ be the type space of all drivers. 

Our decision model \eqref{eq:cost-i-nom} for drivers is similar to the multi-objective optimization frameworks proposed by \cite{yin2006} and \cite{niazi2024}. 
In these frameworks, optimal solutions form a Pareto front where neither emissions nor travel time can be improved without sacrificing the other objective. 
The driver's type parameter $\theta_i$ precisely determines her position on this Pareto front, where $\theta_i$ closer to 1 indicates a stronger preference for emissions reduction over travel time, while $\theta_i$ closer to 0 prioritizes minimizing travel time. 
However, since the TSO doesn't know the types $\theta_i$, $i\in\mc N$, she aims to design a mechanism to elicit this information from the drivers.

\begin{rem}
    Since drivers can easily observe the time to destination through navigation apps like Google Maps, it is natural to assume that they choose a driving policy that only optimizes their travel time. 
    However, similar to \cite{yin2006}, the decision model of drivers considered in \eqref{eq:cost-i-nom} assumes that drivers optimize a combination of both emissions and travel time unless $\theta_i=1$ or $\theta_i=0$.
    This model is valid because many newer vehicles offer a fuel economy estimation feature on the dashboard that displays the number of kilometers/miles per unit of fuel (liter/gallon for internal combustion engine vehicles and kWh for electric vehicles) based on the style with which the vehicle is being driven. 
    The fuel economy feature provides valuable feedback to the driver — essentially showing how far the car could travel on a single unit of fuel if driven in the same manner as currently. 
    Using this feature, drivers become aware of their fuel consumption and adopt driving policies that increase their fuel economy or, equivalently, decrease their emissions \cite{boriboonsomsin2010}.
    \hfill $\diamond$
\end{rem}

The task of the TSO is to design an incentive rate function $u:\Theta\to U$, where $u_i(\theta)\in U_i=\bb R_+$ and $U=\prod_{j=1}^n U_j$, for persuading the drivers to increase their eco-driving levels. 
This incentive transforms the nominal cost of driver~$i$ to an incentivized cost $\ell_i:A\times\Theta_i\times U_i \to \bb R$, which is given by
\begin{equation}
    \label{eq:cost-i-inc}
    \ell_i(a_i,a_{-i},\theta_i, u_i(\theta)) =  c_i(a_i,a_{-i},\theta_i) - u_i(\theta) a_i
\end{equation}
where $c_i$ is given in \eqref{eq:cost-i-nom} and $u_i(\theta) a_i$ is the total incentive received by driver~$i$ for choosing an eco-driving level $a_i$.

We consider traffic conditions in the network where eco-driving leads to extended travel times. 
For instance, maintaining a smooth, consistent speed can reduce emissions but may take longer than aggressive driving with frequent accelerations \cite{coloma2017}. 
To incentivize eco-friendly driving, the TSO compensates drivers for any additional travel time incurred. 
It is important to note that there could be other scenarios, such as eco-driving near signalized intersections \cite{mintsis2020, sun2020, jayawardana2022}, where eco-driving reduces both travel time and emissions. 
However, in these scenarios, drivers do not need incentives to eco-drive according to the emissions-travel time trade-off model considered in \eqref{eq:cost-i-nom} and \eqref{eq:cost-i-inc}.

\subsection{Eco-driving Game induced by EDIM}

The \textit{eco-driving incentive mechanism} (EDIM) is denoted by $(f,u)$, where $f:\Theta\to A$ is the eco-driving recommendation function that the TSO wants the drivers to comply with, and $u:\Theta\to U$ is the incentive rate that the TSO provides to the drivers. The vector $f(\theta)=(f_1(\theta),\dots,f_n(\theta))\in A$ is interpreted as the minimum eco-driving levels recommended to the drivers and the vector $u(\theta)=(u_1(\theta),\dots,u_n(\theta))\in U$ as the incentive rates allocated to the drivers. Driver~$i$ is said to comply with the recommendation $f_i(\theta)$ when she chooses her eco-driving level $a_i\geq f_i(\theta)$, for which she receives the incentive amount equal to $u_i(\theta)a_i$.

\begin{rem}
    In the incentive design literature, $f$ is also known as the \textit{social choice function}. Moreover, when the true types of the drivers are known, the outcome $(f(\theta),u(\theta))$ implemented by the incentive mechanism is called the \textit{first-best solution}, which will be the focus of \cref{sec:mech-known-types}. When the true types are unknown, the implemented outcome is called the \textit{second-best solution}, which will be the focus of \cref{sec:mech-unknown-types}.
    \hfill $\diamond$
\end{rem}

By $(\ell_1,\dots,\ell_n)$, we denote a non-cooperative game, called the \textit{eco-driving game (EDG)}, between $n$ drivers induced by the EDIM $(f,u)$. In this game, given the incentive $u_i(\theta)$, each driver~$i$ optimizes her eco-driving level $a_i\in A_i$ on the route $\mc R_i$, respectively, depending on the eco-driving levels of influencing drivers $j\in\mc N_i$. 

\begin{define}
    The eco-driving profile $a^* = (a_1^*,\dots, a_n^*)$ is a \emph{Nash equilibrium} of the induced EDG $(\ell_1,\dots,\ell_n)$ if for a given $\theta\in\Theta$ and $u\in U$,
    \begin{equation}
        \label{eq:NE-def}
        \!\ell_i(a_i^*,a_{-i}^*, \theta_i, u_i) \!\leq \ell_i(a_i,a_{-i}^*, \theta_i, u_i), ~
        \forall i\in \mc N, \forall a_i\in A_i
    \end{equation}
    where $\ell_i$ is the incentivized cost given in \eqref{eq:cost-i-inc}.
\end{define}

That is, at a Nash equilibrium, no driver can benefit from unilaterally deviating from their equilibrium eco-driving level $a_i^*$. Because of \cref{assume:convex-diff}, the existence of a Nash equilibrium is guaranteed by Debreu~\cite{debreu1952}, 
as restated below.

\begin{lem}
    \label{lem:NE-exists}
    Let \cref{assume:convex-diff} hold. Then, the eco-driving game $(\ell_1,\dots,\ell_n)$ admits a Nash equilibrium.
\end{lem}

Although \cref{lem:NE-exists} guarantees the existence of a Nash equilibrium, it may not be unique. Its uniqueness is guaranteed by the so-called diagonally strict convexity condition by Rosen~\cite{rosen1965}, which can be interpreted as the monotonicity of the vector of gradients of the cost functions $\ell_i$ with respect to $a_i$, for every $i\in\mc N$. This condition is satisfied if the emission $x_i$ and travel time $y_i$ are strictly convex functions.

Note that, in the absence of incentives, Nash equilibria typically do not yield system-level optimal outcomes. Mechanism design theory specifically addresses this gap under incomplete information. In our formulation of the incentive mechanism in \cref{sec:mech-known-types}, the TSO aims to achieve an optimal system-level outcome (i.e., minimum overall emissions) while inducing that outcome as a Nash equilibrium through incentives.

\subsection{Obedience}

We determine the proper incentive $u_i(\theta)$ that ensures that the driver complies with the recommendation $f_i(\theta)$ by choosing $a_i\geq f_i(\theta)$. That is, for every $i\in\mc N$, $u_i(\theta)\in U_i$ is such that it is in the best interest of driver~$i$ to comply with the recommended eco-driving level $f_i(\theta)\in A_i$. 

\begin{define}
    \label{def:IR}
    Given $\theta\in\Theta$, the EDIM $(f,u)$ is said to be \textit{obedient} if for every $i\in\mc N$,
    \begin{multline}
        \label{eq:IR-def}
        c_i(f(\theta),\theta_i) - u_i(\theta)f_i(\theta) \\ \leq c_i(a_i,f_{-i}(\theta),\theta_i) - u_i(\theta) a_i, \quad
        \forall a_i< f_i(\theta).
    \end{multline}
\end{define}

The above inequality \eqref{eq:IR-def} relates to the compliance of drivers with the EDIM's recommendation $f(\theta)$. In the eco-driving setting, the EDIM incentivizes drivers to select eco-driving levels at least as high as the recommended ones. That is, for every $\theta_i\in\Theta_i$, an obedient mechanism $(f,u)$ ensures that it is in the best interest of each driver~$i$ to comply with the recommended eco-driving level $f_i(\theta)\in A_i$, i.e., it is optimal for $i$ to choose her eco-driving level $a_i\geq f_i(\theta)$, given that all other drivers also comply minimally, i.e., $a_{-i}=f_{-i}(\theta)$. This notion of obedience is `ex-post' in nature, like Nash equilibrium, where one has that, given all other drivers choose the minimal recommended eco-driving level $a_{-i}=f_{-i}(\theta)$, it is in the best interest of driver~$i$ also to choose $a_i\geq f_i(\theta)$.

\begin{prop}
    \label{prop:IR-char}
    Given $\theta\in\Theta$, the EDIM $(f,u)$ is obedient if and only if, for every $i\in\mc N$,
    \begin{align}
        u_i(\theta) \!\geq \theta_i \xi_i(a_i,\! f(\theta)) \!+\! (1-\theta_i) \tau_i(a_i,\! f(\theta)), ~
        \forall a_i<\! f_i(\theta) 
        \label{eq:IR-char1}
    \end{align}
    where
    \begin{subequations}
        \label{eq:xi-tau-def}
        \begin{align}
            \xi_i(a_i,f(\theta)) &\eqdef \frac{x_i(f(\theta))-x_i(a_i,f_{-i}(\theta))}{f_i(\theta)-a_i} \\
            \tau_i(a_i,f(\theta)) &\eqdef \frac{y_i(f(\theta))-y_i(a_i,f_{-i}(\theta))}{f_i(\theta)-a_i}.
        \end{align}
    \end{subequations}
\end{prop}

The proof is straightforward by rearranging \eqref{eq:IR-def} and dividing both sides by $f_i(\theta)-a_i>0$.
Since $x_i(a_i,a_{-i})$ is non-increasing in $a_i$ (\cref{assume:emission-decreasing}), we have that $\xi_i(a_i,f(\theta))\leq 0$ for every $a_i< f_i(\theta)$. Therefore, the characterization \eqref{eq:IR-char1} of obedience provides an appropriate amount of incentive for driver~$i$ to ensure her compliance. That is, the minimum amount of incentive that ensures compliance/obedience of driver~$i$ should be larger than the maximum of the weighted sum (according to $\theta_i$) of the rate at which $i$ can increase her emissions $\xi_i(a_i,f(\theta))$ and the rate at which $i$ can decrease her travel time $\tau_i(a_i,f(\theta))$ by not complying.

\section{First-best Eco-driving Incentive Mechanism} 
\label{sec:mech-known-types}

The first-best optimal EDIM is the mechanism that achieves the best outcome when the TSO has perfect knowledge of the drivers' types \cite{laffont2002}. 
Equivalently, one can assume the drivers truthfully report their types to the TSO. 
Then, the goal of the TSO is to minimize the overall emissions by recommending a minimum level of eco-driving $f_1(\theta)\in A_1,\dots,f_n(\theta)\in A_n$ and respectively allocating the incentives $u_1(\theta)\in U_1,\dots,u_n(\theta)\in U_n$ to the drivers. 
In this section, we show that the first-best EDIM is obedient and formulate it as a constrained optimization problem.

\subsection{Obedience under Nash Equilibrium}

We observe that the condition of obedience is naturally satisfied when the EDIM $(f,u)$ implements the recommendation function $f:\Theta\to A$ in Nash equilibrium, i.e., there exists a Nash equilibrium $a^*\in A$ of the induced EDG $(\ell_1,\dots,\ell_n)$ such that $f(\theta)=a^*$.

\begin{lem}
    \label{lem:NE-obedience}
    If the EDIM $(f,u)$ implements $f:\Theta\to A$ in Nash equilibrium, then it is obedient.
\end{lem}

Notice that the implication does not hold in the other direction. 
That is, there could be obedient EDIMs $(f,u)$ that may not implement $f$ in a Nash equilibrium. 
However, in light of \cref{lem:NE-obedience}, it is reasonable to design EDIMs that implement the eco-driving recommendation function $f$ in a Nash equilibrium. 

When the types $\theta_1,\dots,\theta_n$ are known or reported truthfully, the TSO can implement the EDIM $(f,u)$ as follows:
\begin{subequations}
    \label{eq:mech}
    \begin{align}
        & \text{minimize}~\sum_{i=1}^n x_i(f)~
        \text{subject to} 
        \label{eq:mech-cost} \\
        & f \!\in\! \{a^*\!\!\in\! A: \forall i\in \mc N, a_i^*\in \argmin_{a_i\in A_i} \ell_i(a_i, a_{-i}^*, \theta_i, u_i)\} 
        \label{eq:mech-ne} \\
        & u \!\in\! \{\hat u \in U : \sum_{i=1}^n \hat u_i \leq b\}
        \label{eq:mech-bc}
    \end{align}
\end{subequations}
where $b\in \bb R_+$ is the total budget of the TSO. Notice that for a fixed $\theta\in\Theta$ and $u:\Theta\to U$, the right-hand side of \eqref{eq:mech-ne} denotes the set of Nash equilibria.

Even if the existence of a Nash equilibrium is guaranteed by \cref{lem:NE-exists}, computing a Nash equilibrium profile for a continuous game $(\ell_1,\dots,\ell_n)$ is a challenging problem. 
Ratliff et al.~\cite{ratliff2013} present a gradient-based algorithm that converges to a local Nash equilibrium, which is a weaker notion than the Nash equilibrium. 
Mertikopoulos and Staudigl in \cite{mertikopoulos2017, mertikopoulos2019} provide a multi-agent online learning algorithm that converges to a Nash equilibrium under a variational condition on the cost functions. 
Toonsi and Shamma \cite{toonsi2024} provide different conditions under which distributed gradient-based algorithms with higher-order dynamics can converge to a Nash equilibrium.
However, due to \cref{assume:convex-diff} and the specific structure of the incentive mechanism \eqref{eq:mech}, the TSO can induce any Nash equilibrium by allocating incentives $u_1(\theta),\dots,u_n(\theta)$ to drivers subject to the budget constraint \eqref{eq:mech-bc}. 
Therefore, in this case, the TSO does not need to compute a Nash equilibrium; rather, he chooses one and implements it.
In other words, any recommendation $f(\theta)$ can be implemented as a Nash equilibrium by appropriately incentivizing drivers subject to the budget constraint.

\begin{rem}
    While our framework focuses on emissions reduction, it is worth noting that energy consumption, which is directly proportional to emissions across vehicle types, can also be considered. 
    For conventional internal combustion engine (ICE) vehicles, emissions are directly proportional to fuel consumption with vehicle-specific emission factors \cite{liu2016}. 
    For electric vehicles, this relationship is mediated by grid emission factors \cite{kawamoto2019}. 
    Recent studies have validated that eco-driving strategies optimized for emissions reduction also yield proportional improvements in energy efficiency \cite{huang2018}, with a high correlation across diverse driving conditions \cite{scora2006}. 
    Our framework can accommodate either objective function with minimal modification, as the mathematical structure of the incentive mechanism remains unchanged whether optimizing for emissions or energy consumption. 
    However, to align with climate policy objectives and regulatory frameworks that explicitly target greenhouse gas emissions, we consider emission reduction as the overarching goal of the TSO.
    \hfill $\diamond$
\end{rem}

\subsection{Implementation in Nash Equilibrium}

As shown in \cref{lem:NE-obedience}, when the recommendation $f(\theta)$ is a Nash equilibrium of the induced EDG $(\ell_1,\dots,\ell_n)$, then the EDIM $(f,u)$ is obedient. In other words, it is in the best interest of each driver~$i$ to comply with the recommended eco-driving level by choosing $a_i\geq f_i(\theta)$ when other drivers choose $a_j= f_j(\theta)$ for all $j\in\mc N \setminus \{i\}$. Therefore, we explore the implementation of EDIM in the Nash equilibrium by appropriately choosing the incentives.

When the drivers truthfully report their types $\theta\in\Theta$, the mechanism $(f,u)$ given by \eqref{eq:mech} is well-posed and it implements the eco-driving recommendation function $f:\Theta\to A$ in Nash equilibrium.
To be precise, subject to \cref{assume:convex-diff}, \cref{lem:NE-exists} proves the existence of a Nash equilibrium of the induced EDG $(\ell_1,\dots,\ell_n)$. Fix $u:\Theta\to U$ such that \eqref{eq:mech-bc} is satisfied. Define $A^*\subseteq A$ as
\begin{equation}
    \label{eq:A-star}
    A^* \!\eqdef\! \{a^*\!\in\! A: \forall i\in\mc N, a_i^*\!\in \argmin_{a_i\in A_i} \ell_i(a_i,a_{-i}^*,\theta_i,u_i)\!\}.
\end{equation}
Then, for every $a^*\in A^*$, we have for every $i\in\mc N$,
\[
\ell_i(a_i^*,a_{-i}^*,\theta_i,u_i) \leq \ell_i(a_i,a_{-i}^*,\theta_i,u_i), \quad \forall a_i\in A_i.
\]
Therefore, $A^*$ is a set of all Nash equilibria of the induced game $(\ell_1,\dots,\ell_n)$.
Finally, from \eqref{eq:mech-ne}, we have that there exists a Nash equilibrium $a^*\in A^*$ of the induced game $(\ell_1,\dots,\ell_n)$ such that $f(\theta)=a^*$.

\subsection{Equivalent Formulation}

Suppose the incentivized cost $\ell_i$ of each driver is convex and differentiable in her own decision variable $a_i$. In that case, a Nash equilibrium of EDG $(\ell_1,\dots,\ell_n)$ exists, and the TSO can implement the EDIM $(f,u)$ in Nash equilibrium, provided that he knows the true types of the drivers. 
The questions that remain are: which Nash equilibrium should the TSO implement, and how to solve \eqref{eq:mech}?
To answer the first question, the TSO should choose the Nash equilibrium that minimizes the overall emissions of the transportation network subject to the budget constraint, i.e., the Nash equilibrium that fits his objective. 
To answer the second question, first notice that solving \eqref{eq:mech} directly is challenging because it is a bilevel optimization problem. 
However, we observe that \eqref{eq:mech} admits an equivalent formulation as a constrained optimization problem, which is straightforward to solve numerically.


\begin{thm}
    \label{thm:mech-equiv}
    Let \cref{assume:convex-diff} and \cref{assume:emission-decreasing} hold, and suppose the drivers truthfully report their types $\theta\in\Theta$. Then, the incentive mechanism $(f,u)$ given in \eqref{eq:mech} can be equivalently formulated as follows:
    \begin{subequations}
        \label{eq:mech-equiv}
        \begin{align}
            f(\theta) \in & \argmin_{a\in A} \sum_{i=1}^n x_i(a) 
            \label{eq:mech-equiv-cost} \\
            & \text{subject to}~ \sum_{i=1}^n |\D_{a_i} c_i(a,\theta_i)|_+ \leq b
            \label{eq:mech-equiv-bc}
        \end{align}
    \end{subequations}
    and for every $i\in\mc N$,
    \begin{equation}
        \label{eq:incentive-equiv}
        u_i(\theta) = |\D_{a_i} c_i(f(\theta),\theta_i)|_+.
    \end{equation}
\end{thm}
\begin{proof}
    Define $A^*\subseteq A$ as in \eqref{eq:A-star}. Then, we prove the result by showing that \eqref{eq:mech-ne}, $f(\theta)\in A^*$, is equivalent to choosing the incentive $u_i(\theta)=|\D_{a_i} c_i(f(\theta),\theta_i)|_+$ for every $i\in\mc N$. This converts the bilevel optimization problem \eqref{eq:mech} to a constrained optimization problem \eqref{eq:mech-equiv}.

    First, we address a degenerate case where $\D_{a_i}c_i(a_i,f_{-i}(\theta),\theta_i) \leq 0$ for every $a_i\in\text{int}(A_i)$. Notice that by \cref{assume:convex-diff}, the nominal cost $c_i(a_i,a_{-i},\theta_i)$ is convex and differentiable with respect to $a_i\in A_i$. Therefore, for every $a_{-i}\in A_{-i}$ and $a_i,\hat a_i\in A_i$ such that $a_i\geq \hat a_i$, we have that $\D_{a_i} c_i(a_i,a_{-i},\theta_i) \geq \D_{a_i} c_i(\hat a_i,a_{-i},\theta_i)$. Thus, if $\D_{a_i}c_i(a_i,a_{-i},\theta_i) \leq 0$ for $a_i\to 1^-$, then $a_i=1$ is the minimizing solution to $\ell_i(a_i,a_{-i},\theta_i,u_i)$ for any $u_i\in U_i$. In such a case, opting for eco-driving level $a_i=1$ is the optimal choice for driver~$i$ and the TSO does not need to incentivize her, i.e., $u_i(\theta)=0$. Moreover, we have $f_i(\theta)=1$ because, by \cref{assume:emission-decreasing}, the emission function $x_i(a_i,a_{-i})$ is non-increasing in $a_i$, i.e., $x_i(1,a_{-i})\leq x_i(a_i,a_{-i})$ for every $a_i\in A_i=[0,1]$.

    In the remaining proof, we consider the case where $\D_{a_i} c_i(a_i,f_{-i}(\theta),\theta_i)>0$ for some $a_i\in A_i$.
    
    Let $f(\theta)\in A^*$, then for every $i\in\mc N$,
    \[
    \ell_i(f(\theta),\theta_i,u_i) \leq \ell_i(a_i,f_{-i}(\theta),\theta_i,u_i), \quad \forall a_i\in A_i.
    \]
    Since $\ell_i$ is convex and differentiable in $a_i$, the stationarity condition $\D_{a_i} \ell_i(f_i(\theta),f_{-i}(\theta),\theta_i,u_i)=0$ is sufficient for showing $f_i(\theta)$ is the minimum given that $a_{-i}=f_{-i}(\theta)$, i.e.,
    \[
    \D_{a_i} c_i(f_i(\theta),f_{-i}(\theta),\theta_i) - u_i(\theta) = 0
    \]
    which gives \eqref{eq:incentive-equiv}. 
    Now, let \eqref{eq:incentive-equiv} hold. Then, for every $i\in\mc N$, the stationarity condition $\D_{a_i} \ell_i(a_i,f_{-i}(\theta),\theta_i,u_i)=0$ holds if $a_i=f_i(\theta)$. Therefore, $f(\theta)\in A^*$.
\end{proof}

The equivalent formulation \eqref{eq:mech-equiv} and \eqref{eq:incentive-equiv} of the incentive mechanism \eqref{eq:mech} reiterates our claim that the TSO can implement any recommendation that minimizes the overall emissions as a Nash equilibrium, provided that he has a sufficient budget to allocate to the drivers. 
That is, the recommended eco-driving levels $f(\theta)$ minimize the overall emissions subject to the budget constraint.

Although \cref{lem:NE-obedience} has already shown that implementing EDIM in the Nash equilibrium ensures obedience, we further emphasize this by showing that the incentive $u_i(\theta)$ chosen as in \eqref{eq:incentive-equiv} satisfies the obedience condition \eqref{eq:IR-char1}. 
For every $i\in\mc N$, $\theta\in\Theta$, and $a_i\geq f_i(\theta)$, it holds that
\begin{multline*}
    u_i(\theta) = |\D_{a_i} c_i(f(\theta),\theta_i)|_+ 
    \geq \D_{a_i} c_i(f(\theta),\theta_i) \\
    \geq \theta_i \xi_i(a_i,f(\theta))+(1-\theta_i)\tau_i(a_i,f(\theta))
\end{multline*}
where $\xi_i$ and $\tau_i$ are given in \eqref{eq:xi-tau-def} and the last step in the above inequality is due to the convexity of the nominal cost $c_i$ (\cref{assume:convex-diff}), i.e., $\D_{a_i}c_i(f_i(\theta),f_{-i}(\theta),\theta_i) \geq \D_{a_i}c_i(a_i,f_{-i}(\theta),\theta_i)$ for every $a_i< f_i(\theta)$.

While the first-best EDIM provides theoretical insights under the idealized assumption that the TSO knows drivers' true types, this assumption rarely holds in practice. 
Without it, the mechanism \eqref{eq:mech} cannot guarantee compliance with the recommended eco-driving levels or achieve any emission reductions. 
This limitation motivates our development in the next section of incentive-compatible mechanisms that remain effective even when drivers strategically misreport their types.

\section{Second-best Eco-driving Incentive Mechanism}
\label{sec:mech-unknown-types}

In this section, we consider the case when the drivers may strategically report their types to the TSO. 
The TSO elicits the types $\theta_i$ from each driver~$i$, who may instead report $\hat\theta_i\in\Theta_i$. 
The reported type $\hat\theta_i$ may not be equal to the true type $\theta_i$ because by reporting her type, driver~$i$ aims to minimize her incentivized cost $\ell_i$ by maximizing her incentive $u_i(\hat\theta_i,\theta_{-i})$.

\subsection{Truthfulness}
In addition to minimizing emissions, the TSO aims to make the EDIM $(f,u)$ resilient to the strategic reporting of different types. 
That is, the TSO designs the eco-driving recommendation function $f:\Theta\to A$ and the incentive function $u:\Theta\to U$ such that the drivers do not gain anything by reporting their types untruthfully. 
In particular, reporting $\hat\theta_i=\theta_i$ minimizes $i$'s incentivized cost $\ell_i(f(\hat\theta_i,\theta_{-i}),\theta_i,u_i(\hat\theta_i,\theta_{-i}))$ under the EDIM $(f,u)$. 
Consideration of this additional requirement, known as the truthfulness constraint, renders the EDIM to be the second-best optimal.

\begin{define}
    \label{def:IC}
    The EDIM $(f,u)$ is called \emph{truthful} if, for every $i\in\mc N$, $\theta_{-i}\in\Theta_{-i}$, and $\theta_i,\hat\theta_i\in\Theta_i$, the incentivized cost $\ell_i(f(\theta),\theta_i,u_i(\theta))=c_i(f(\theta),\theta_i) - u_i(\theta) f_i(\theta)$ given in \eqref{eq:cost-i-inc} satisfies
    \begin{multline}
        \label{eq:IC-def}
        c_i(f(\theta_i,\theta_{-i}),\theta_i) - u_i(\theta_i,\theta_{-i}) f_i(\theta_i,\theta_{-i}) \\ 
        \leq c_i(f(\hat\theta_i,\theta_{-i}),\theta_i) - u_i(\hat\theta_i,\theta_{-i}) f_i(\hat\theta_i,\theta_{-i}).
    \end{multline}
\end{define}

Truthfulness states that, for all drivers~$i\in\mc N$ and all type profiles $(\theta_i,\theta_{-i})\in\Theta$, given that other drivers report their types truthfully, it is in the best interest of driver~$i$ also to report her type truthfully. Under this notion, reporting truthfully is a Nash equilibrium of a type-reporting game induced by the EDIM $(f,u)$ in which each driver reports $\hat\theta_i\in\Theta_i$ to the TSO to minimize her nominal cost $c_i$ while maximizing her incentive $u_i f_i$.

\begin{prop}
    \label{prop:IC-char}
    Let \cref{assume:convex-diff} hold. Then, the eco-driving incentive mechanism $(f,u)$ is truthful if and only if, for every $i\in\mc N$, $\ell_i(f(\theta),\theta_i,u_i)$ is concave with respect to $\theta_i$ and, for every $\theta_i\in\Theta_i$ and every $\theta_{-i}\in\Theta_{-i}$,
    \begin{multline}
        \label{eq:ic-char}
        \D_{\theta_i} \!\ell_i(f(\theta_i,\theta_{-i}),\theta_i,u_i(\theta_i,\theta_{-i})) \\ = x_i(f(\theta_i,\theta_{-i})\!) - y_i(f(\theta_i,\theta_{-i})\!).
    \end{multline}
\end{prop}

\cref{prop:IC-char}, proved in \cref{appendix:prop:IC}, has several implications in terms of designing an incentive-compatible EDIM. 
An immediate consequence of \eqref{eq:ic-char} is that the following differential equation must be satisfied:
\begin{equation*}
    \D_{\theta_i} \!\big(u_i(\theta)\!\cdot\! f_i(\theta)\big) \!=\! \theta_i\nabla_f x_i^\T \D_{\theta_i} f(\theta) + (1-\theta_i) \nabla_f y_i^\T \D_{\theta_i} f(\theta).
\end{equation*}
It can be observed that if the incentive $u_i$ and the recommendation $f$ do not depend on the reported type $\hat\theta_i$, for every $i\in\mc N$, then the above equation holds.


\subsection{Implementation}

When the drivers report their types strategically, then the EDIM $(f,u)$ is incentive compatible\footnote{Here, we mean the \textit{ex-post} incentive compatibility: Given that $-i$ report $\theta_{-i}$ truthfully and comply with the recommendation $f_{-i}$, each driver $i$ finds it optimal to report $\theta_i$ truthfully and comply with the recommendation $f_i$.} if it is truthful and obedient \cite{bonatti2024}. 
Consider an EDIM $(f,u)$, which will be referred to as \textit{second-best}, where
\begin{subequations}
    \label{eq:mech-implementable}
    \begin{align}
        f \in & \argmin_{a\in A} \sum_{i=1}^n x_i(a) 
        \label{eq:mech-implementable-cost} \\
        & \text{subject to}~ \sum_{i=1}^n |\D_{a_i}y_i(a_i,a_{-i})|_+ \leq b
        \label{eq:mech-implementable-bc}
    \end{align}
\end{subequations}
and for every $i\in\mc N$,
\begin{equation}
    \label{eq:incentive-implementable}
    u_i = |\D_{a_i} y_i(f_i,f_{-i})|_+.
\end{equation}

\begin{thm}
    Let \cref{assume:convex-diff} and \cref{assume:emission-decreasing} hold. Then, the EDIM given by \eqref{eq:mech-implementable}-\eqref{eq:incentive-implementable} is truthful and obedient.
\end{thm}
\begin{proof}
    Firstly, we show that choosing the incentive as in \eqref{eq:incentive-implementable} satisfies the condition of obedience \eqref{eq:IR-char1}. For every $a_i< f_i(\theta)$, it holds that 
    \begin{align*}
        u_i = |\D_{a_i} y_i(f)|_+
        & \geq \D_{a_i} y_i(f) \\
        & \overset{\text{(a)}}{\geq} \tau_i(a_i,f(\theta)) \\
        & \overset{\text{(b)}}{\geq} (1-\theta_i)\tau_i(a_i,f(\theta)) \\
        & \overset{\text{(c)}}{\geq} \theta_i \xi_i(a_i,f(\theta)) + (1-\theta_i) \tau_i(a_i,f(\theta))
    \end{align*}
    where $\xi_i$ and $\tau_i$ are given in \eqref{eq:xi-tau-def}, and (a) is because of the convexity of $y_i$ with respect to $a_i$, (b) is because $\theta_i\in[0,1]$, and (c) is because of \cref{assume:emission-decreasing} which implies $\xi_i(a_i,f(\theta))\leq 0$ for every $a_i< f_i(\theta)$. 
    
    Secondly, we show that the EDIM \eqref{eq:mech-implementable}-\eqref{eq:incentive-implementable} is truthful by showing that the conditions of \cref{prop:IC-char} hold. Notice that \eqref{eq:mech-implementable} and \eqref{eq:incentive-implementable} do not depend on $\theta_i$. Thus, $\D_{\theta_i} u_i(\theta)=0$ and $\D_{\theta_i} f(\theta) = 0_n$, implying
    \begin{align*}
        \D_{\theta_i} \ell_i(f,\theta_i,u_i) & = x_i(f) + \theta_i \nabla_f x_i^\T \D_{\theta_i} f - y_i(f) \\
        & \quad\quad + (1-\theta_i) \nabla_f y_i^\T \D_{\theta_i} f - \D_{\theta_i} (u_i\cdot f_i) \\
        &= x_i(f) - y_i(f).
    \end{align*}
    Moreover, $\ell_i(f,\theta_i,u_i) = \theta_i x_i(f) + (1-\theta_i) y_i(f) - u_i f_i$ is linear in $\theta_i$, thus it is concave. This concludes the proof.
\end{proof}

Under the first-best EDIM, each driver is compensated according to the rate of increase in her cost $c_i$ at the recommendation $f\in A$. 
This is because the drivers are assumed to report truthfully, which allows the TSO to compute $\D_{a_i} c_i(f(\theta),\theta_i)$.
The second-best EDIM, assuming that drivers report strategically, compensates each driver according to the rate of increase in her travel time $y_i$ at the recommendation $f\in A$. 
The incentive in the second-best EDIM is larger than the incentive in the first-best EDIM, i.e., for every $a\in A$, $i\in\mc N$, and $\theta_i\in\Theta_i$,
\begin{multline*}
    |\D_{a_i} y_i(a)|_+ \geq (1-\theta_i) |\D_{a_i} y_i(a)|_+ \\
    = \theta_i|\D_{a_i}x_i(a)|_+ + (1-\theta_i)|\D_{a_i} y_i(a)|_+
    \geq |\D_{a_i} c_i(a,\theta_i)|_+
\end{multline*}
where we used $\D_{a_i}x_i(a) \leq 0$ because of \cref{assume:emission-decreasing}. 

From the discussion above, it might appear that enforcing the truthfulness constraint comes at a cost. 
That is, given a fixed total budget $b\in\bb R_+$, one may argue that the TSO can achieve higher eco-driving levels in the first-best EDIM than in the second-best EDIM. 
However, this holds only when the drivers report truthfully. 
When drivers report strategically, the first-best EDIM may no longer be obedient, and drivers might find lower eco-driving levels than the recommended $f$ optimal at equilibrium. 
Consequently, the first-best EDIM can result in higher overall emissions compared to the second-best under strategic reporting.

Another observation is that the second-best EDIM \eqref{eq:mech-implementable}-\eqref{eq:incentive-implementable} does not implement $f$ as a Nash equilibrium in the sense that $a=f$ is not a Nash equilibrium. 
However, since the second-best EDIM is obedient and the EDG it induces has a Nash equilibrium (\cref{lem:NE-exists}), the eco-driving profiles at Nash equilibria are higher than the recommendation $f$. 
That is, for every $a^*\in A^*$, where $A^*$ is a set of Nash equilibria given by \eqref{eq:A-star}, it holds that $a_i^*\geq f_i$ for every $i\in\mc N$. 
Therefore, the second-best EDIM implements $f$ in a Nash equilibrium if we refine the definition of implementation in a Nash equilibrium. 
Under such a definition, the mechanism guarantees that any equilibrium action profile will result in eco-driving levels at least as high as the recommended levels $f\in A$.


\section{Numerical Simulations}
\label{sec:simulation}

This section presents simulations designed to demonstrate the theoretical properties of our proposed eco-driving incentive mechanisms. We illustrate how the framework conceptually captures strategic behavior and driver interactions. We provide a clear proof-of-concept demonstration on the effects of budget constraints, driver types, and strategic misreporting on the equilibrium outcomes and overall emissions reduction.

\subsection{Simulation Setup}
Consider $n=10$ drivers in a transportation network who interact with each other according to a weighted interaction matrix $W=[w_{ij}]_{i,j=1,\dots,n}$, where $w_{ii}=1$ and $w_{ij}\in[0,1]$ is the weight with which the eco-driving level $a_j$ of driver~$j$ influences the emissions and travel time of driver~$i$ in expectation along her route $\mc R_i$. 
For simulation purposes, we consider the emission function of each driver~$i$ to be
\begin{equation}
    \label{eq:em-fun}
    x_i(a_i,a_{-i}) = \ol x_i \alpha_i^{a_i + \sum_{j\in\mc N_i} w_{ij} a_j}
\end{equation}
where $\alpha_i\in(0,1)$, and $\ol x_i\in\bb R_+^*$ is the maximum emissions when $a=0_n$. 
This emissions function, inspired by \cite{yin2006}, allows drivers to reduce emissions by increasing eco-driving compliance. It also captures driver interactions: when surrounding vehicles maintain higher eco-driving levels $a_{-i}$, driver~$i$ achieves lower emissions at any given $a_i$. This reflects how eco-driving vehicles create smoother traffic flow with fewer disruptions, \cite{barth2009, yang2016}, enabling individual drivers to reduce emissions more effectively while maintaining traffic harmony, thereby decreasing network-wide emissions.

Consider the travel time function of each driver~$i$ to be
\begin{equation}
    \label{eq:tt-fun}
    y_i(a_i,a_{-i}) \!=\! \beta_i\!\!\left(\! a_i - \frac{\sum_{j\in\mc N_i}\! w_{ij} a_j}{\sum_{j\in\mc N_i}\! w_{ij}} \!\right)^2 \!\!+ \gamma_i\!\! \sum_{j\in\mc N_i}\! w_{ij} a_j + \ul y_i
\end{equation}
where $\beta_i,\gamma_i\in\bb R_+$, and $\ul y_i\in\bb R_+^*$ is the minimum travel time when $a=0_n$. 
This travel time function captures a key trade-off: when surrounding drivers adopt higher eco-driving levels, driver~$i$ experiences increased travel times unless she similarly increases her own eco-driving level. 
This reflects how eco-driving vehicles typically maintain slower, steadier speeds that may constrain non-eco-driving vehicles in traffic. 
Both emission and travel time functions satisfy our mathematical requirements (\cref{assume:convex-diff} and \ref{assume:emission-decreasing}): the emission function is convex, differentiable, and decreasing in $a_i$, while the travel time function is convex and differentiable with respect to $a_i$.

In \eqref{eq:em-fun} and \eqref{eq:tt-fun}, the interaction weight $w_{ij}$, for every $i,j\in\mc N$ and $j\neq i$, is chosen to be $0$ with probability $0.5$ and uniformly randomly in $[0,1]$ with probability $0.5$. 
Denoting $\mc U[a,b]$ to be the uniform distribution over $[a,b]$, for some real numbers $a\leq b$, we choose $\alpha_i\sim\mc U[0.6,0.8]$, $\beta_i\sim\mc U[2,3]$, and $\gamma_i\sim\mc U[3,4]$ for every $i\in\mc N$. 
Moreover, we let $\ol x_i=4$ and $\ul y_i=1$, for every $i\in\mc N$.
Similarly, for every $i\in\mc N$, the type $\theta_i\sim\mc U[0,0.4]$ of driver~$i$, which means that in general drivers put more weight $1-\theta_i$ on minimizing their travel times than minimizing their emissions. 
We consider the first-best EDIM $(f,u)$ given by \eqref{eq:mech-equiv} and \eqref{eq:incentive-equiv}, and the second-best EDIM $(f,u)$ given by \eqref{eq:mech-implementable} and \eqref{eq:incentive-implementable}. 
We use MATLAB's Optimization Toolbox to solve \eqref{eq:mech-equiv} and \eqref{eq:mech-implementable}.

\subsection{Obedience and Truthfulness}

Assuming that all other drivers truthfully report their types, we illustrate the outcome when driver~$i$, $i=1$, misreports her type $\hat\theta_i\neq \theta_i$. Suppose the total budget is $b=3$. \cref{fig:FB-levels} shows the result obtained by the first-best EDIM in terms of the recommended eco-driving level $f_i(\hat\theta_i,\theta_{-i})$ and the optimal eco-driving level
\[
a_i^{\mathrm{opt}} = \argmin_{a_i\in A_i} \ell_i(a_i,f_{-i}(\hat\theta_i,\theta_{-i}),\theta_i,u_i(\hat\theta_i,\theta_{-i}))
\]
where, for illustration purposes, we assume that driver~$i$ knows $f_{-i}(\hat\theta_i,\theta_{-i})$ so that the above optimization problem is well-posed. 
If driver~$i$ underreports her type, $\hat\theta_i<\theta_i$, then we see that the first-best EDIM ensures obedience, i.e., $a_i^{\mathrm{opt}}\geq f_i$. However, when driver~$i$ overreports her type, $\hat\theta_i>\theta_i$, then the first-best EDIM is not obedient anymore because $a_i^{\mathrm{opt}}< f_i$. This is because overreporting one's type results in higher recommended eco-driving levels that contradict the driver's actual preferences, making compliance suboptimal.

On the other hand, the second-best EDIM is obedient, as illustrated in \cref{fig:SB-levels}, because the incentive amount $u_i$ in \eqref{eq:incentive-implementable} is sufficiently large and is indifferent to the reported type $\hat\theta_i$. 
However, observe that both the recommended and optimal eco-driving levels of driver~$i$ are, in general, lower under the second-best EDIM than the first-best. 
This is because enforcing the truthfulness constraint makes the mechanism conservative. 
Moreover, the gap between the eco-driving levels of the first-best and the second-best is huge because only one driver is misreporting her type. 
If the truthfulness is not enforced and all drivers misreport their types, the performance of the first-best EDIM will worsen. 

\begin{figure}[!t]
    \centering
    \subfloat[First-best EDIM is not obedient when driver~$i$ reports $\hat\theta_i$ greater than the true $\theta_i$. \label{fig:FB-levels}]{\includegraphics[width=0.7\linewidth]{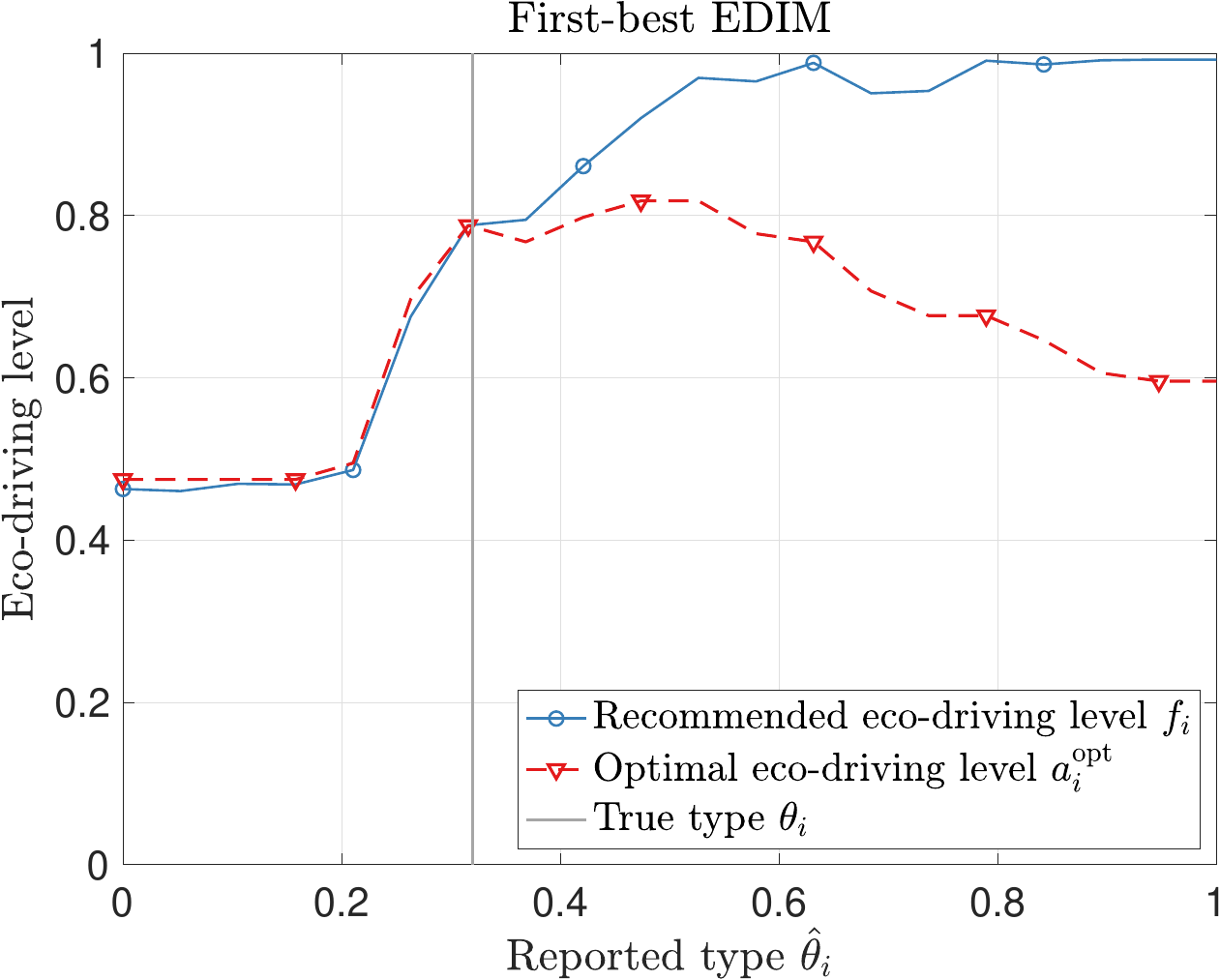}}

    \subfloat[Second-best EDIM is always obedient irrespective of the reported $\hat\theta_i$. \label{fig:SB-levels}]{\includegraphics[width=0.7\linewidth]{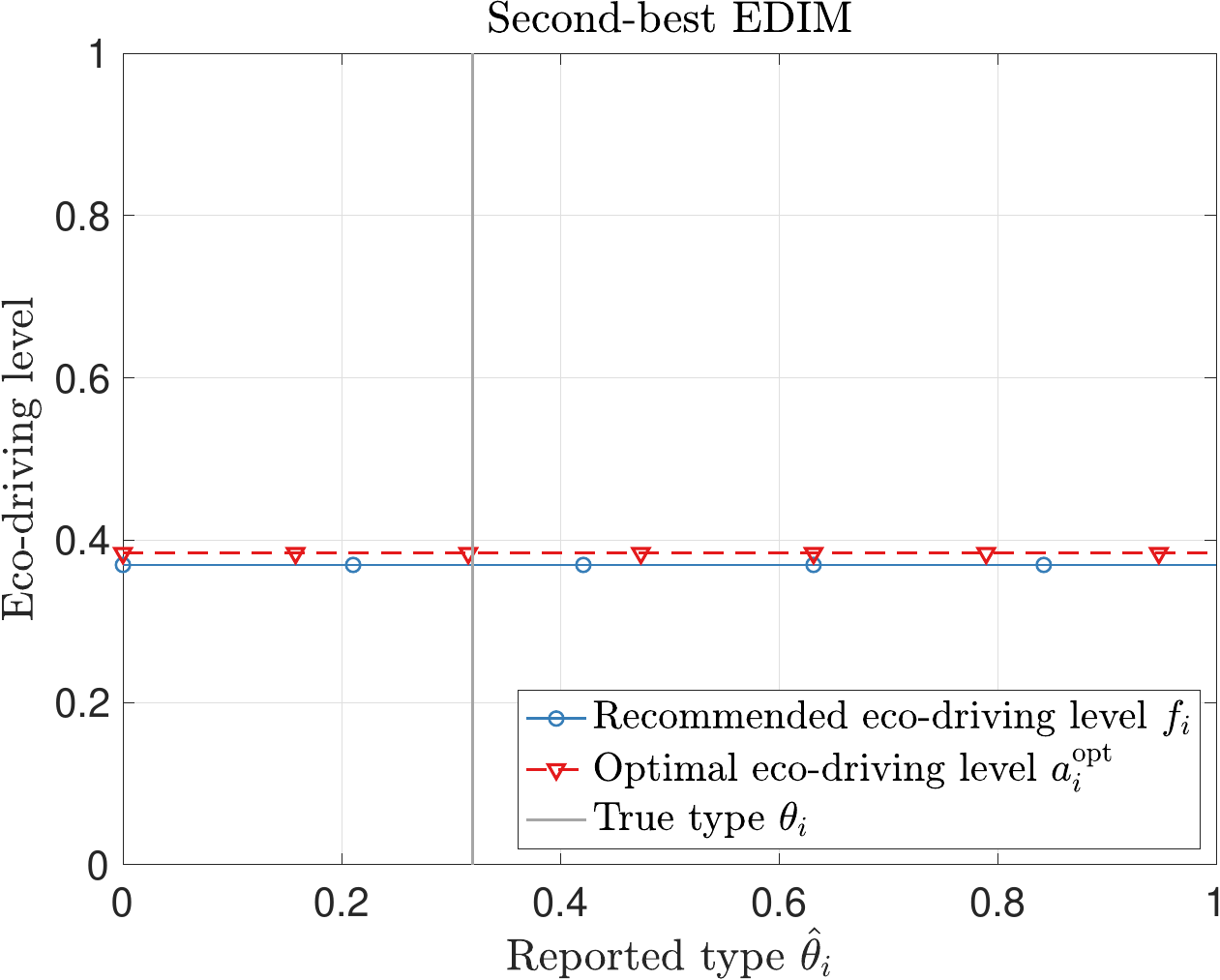}}
    
    \caption{Recommended and optimal eco-driving levels as functions of the reported type $\hat\theta_i$.}
    \label{fig:levels}
\end{figure}

\begin{figure}[!t]
    \centering
    \subfloat[First-best EDIM is not truthful as it is optimal for driver~$i$ to report $\hat\theta_i$ greater than the true $\theta_i$. \label{fig:FB-cost}]{\includegraphics[width=0.985\linewidth]{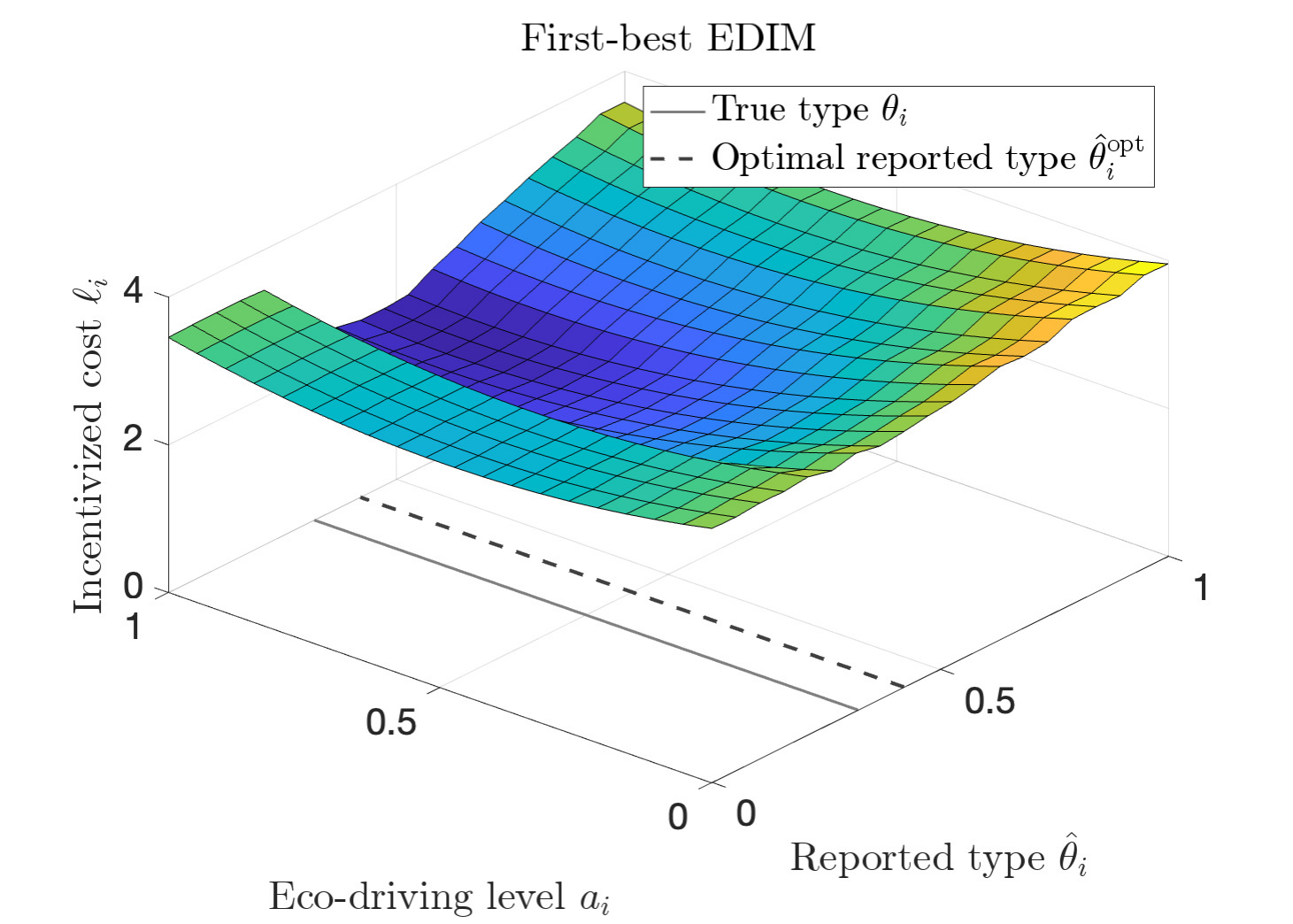}}

    \subfloat[Second-best EDIM is truthful as driver~$i$ does not gain anything by reporting untruthfully. \label{fig:SB-cost}]{\includegraphics[width=0.985\linewidth]{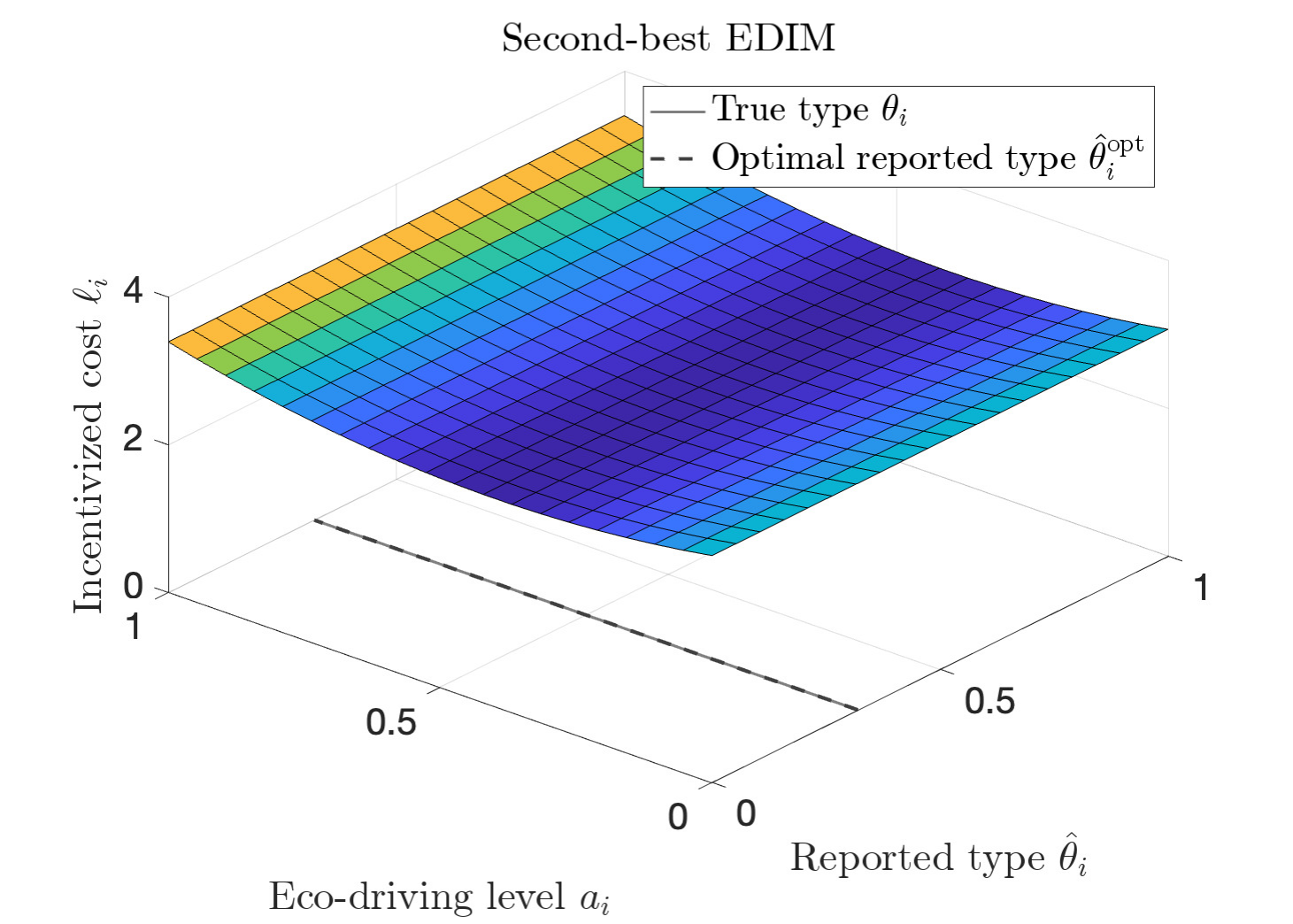}}
    
    \caption{Incentivized cost $\ell_i(a_i,f_{-i}(\hat\theta_i,\theta_{-i}),\theta_i,u_i(\hat\theta_i,\theta_{-i}))$ as a function of eco-driving level $a_i$ and reported type $\hat\theta_i$.}
    \label{fig:cost}
\end{figure}

To demonstrate truthfulness, we illustrate the incentivized cost $\ell_i$ as a function of driver~$i$'s reported type $\hat\theta_i$ and eco-driving level $a_i$ in \cref{fig:cost}, where we assume that all other drivers report truthfully $\hat\theta_{-i}=\theta_{-i}$ and comply with the recommendation $a_{-i}=f_{-i}(\hat\theta_i,\theta_{-i})$. We see that under the first-best EDIM, driver~$i$ may find it optimal to overreport her type to minimize her cost $\ell_i$, whereas the second-best EDIM ensures that misreporting does not gain anything for driver~$i$. In fact, truthful reporting is a Nash equilibrium strategy for $i$.

\subsection{Overall Emissions Reduction with Increased Budget}

\begin{figure}[!t]
    \centering
    \includegraphics[width = 0.82\linewidth]{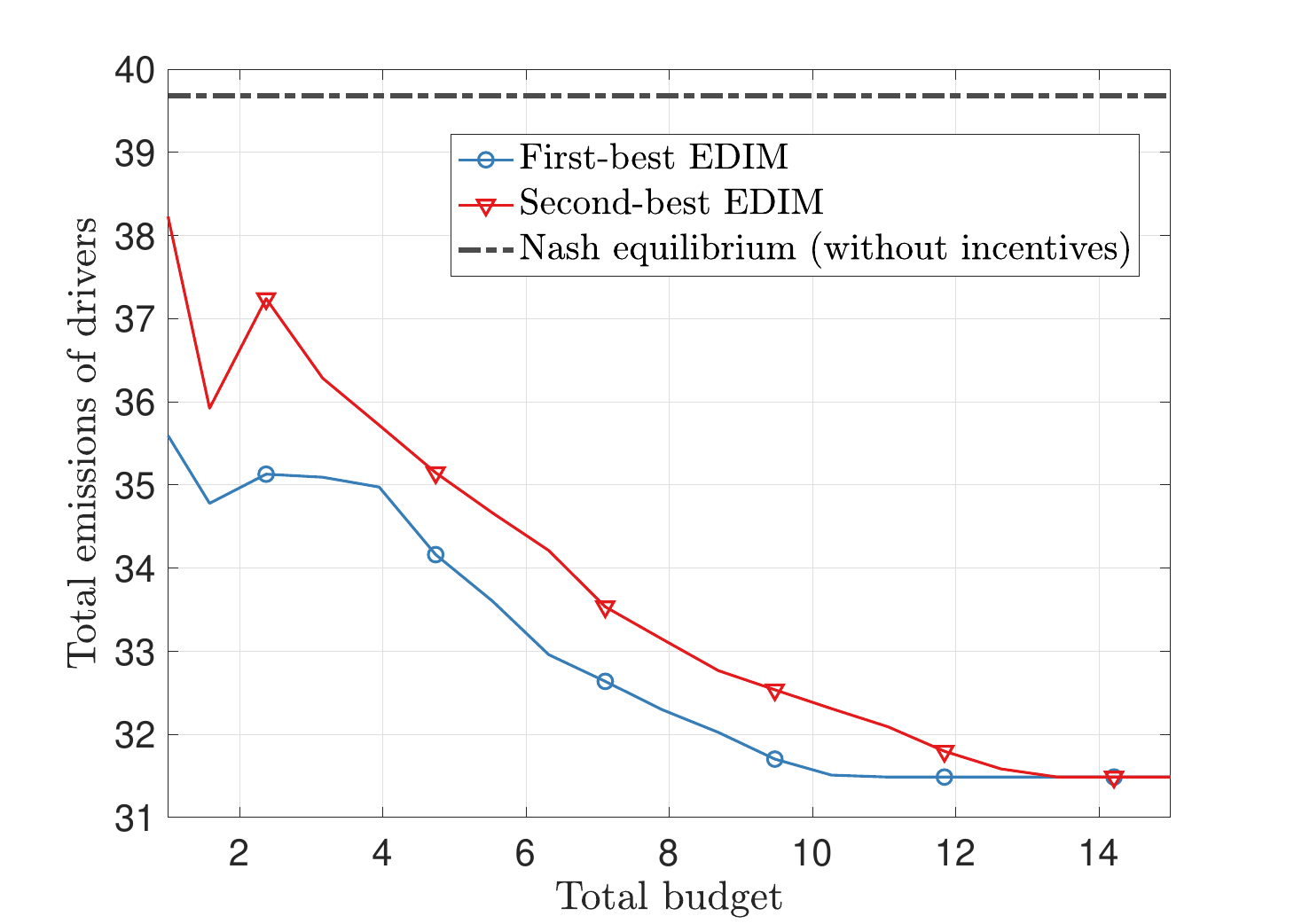}
    \caption{Emissions of drivers as a function of the budget.}
    \label{fig:em-tt-bud}
\end{figure}

We demonstrate the effect of the total budget $b$ on the overall emissions of the network, assuming that drivers truthfully report their types. 
\cref{fig:em-tt-bud} shows the total emissions of drivers under different values of the total budget and compares it to the case when there are no incentives and drivers choose their eco-driving levels based on the Nash equilibrium.
For the incentive mechanisms, there is a decreasing trend in emissions because, with a higher budget, the TSO can implement higher eco-driving levels under both first-best and second-best EDIMs, resulting in lower emissions.

Since the first-best EDIM results in higher eco-driving levels than the eco-driving levels of the second-best, we observe in \cref{fig:em-tt-bud} that with the same budget, the first-best EDIM in general achieves lower emissions than the second-best. 
Moreover, notice that for a budget $b>7$, the emissions under the first-best EDIM converge to a constant value. 
This is because, with a larger budget, the TSO can, in principle, incentivize all the drivers to choose maximum eco-driving levels $a=f(\theta)=1_n$, where they fully comply with the eco-driving guidance. 
On the other hand, under the second-best EDIM, $a=f(\theta)=1_n$ is achieved for budget $b>9$. 

To conclude, while the first-best EDIM achieves lower emissions than the second-best EDIM when drivers are truthful, it is shown to violate the obedience constraint when even a single driver misreports her type. 
Moreover, under the first-best EDIM, it is optimal for drivers to misreport their types because their recommendations and incentives depend on their respective reported types.
Conversely, the second-best EDIM is both obedient and truthful, as drivers not only comply with the recommended eco-driving levels but also find it optimal to report their types truthfully.
If the TSO has a sufficiently large budget, he can incentivize drivers to achieve maximum eco-driving levels and ensure their full compliance. 
However, under the second-best EDIM, the budget required to ensure full compliance is larger than that required by the first-best EDIM with known types. 
Therefore, in addition to the emission cost demonstrated in \cref{fig:em-tt-bud}, the TSO also incurs additional budget costs to incentivize similar levels of eco-driving when drivers strategically report their types.

\section{Concluding Remarks}
\label{sec:conclusion}

This paper introduces a theoretical framework for incentive mechanisms to promote eco-driving with the overarching goal of minimizing emissions in urban transportation networks.
The timeline of the incentive mechanisms is as follows: 1)~the TSO elicits preferences or types from the drivers, 2)~the TSO computes eco-driving levels that minimize the overall emissions subject to the limited budget which is to be optimally allocated as incentives to the drivers, and 3)~the TSO provides each driver with the recommended eco-driving level that she needs to comply with and offers an incentive rate that shapes her cost function.
When drivers truthfully report their types, the first-best EDIM implements the recommended eco-driving levels at the Nash equilibrium. This implies obedience, as drivers find it optimal to choose their eco-driving levels at least as high as the recommended levels. 
However, when the drivers strategically report their types, the obedience constraint may not hold. 
Therefore, the second-best EDIM is proposed, which, in addition to ensuring obedience constraint, also incorporates the truthfulness constraint, guaranteeing that the drivers do not gain any advantage by misreporting their types. Since the second-best EDIM is robust, it is conservative and generally requires more budget to achieve similar outcomes as the first-best EDIM.
However, it must be clarified that the first-best EDIM achieves lower emissions than the second-best only when the drivers are truthful. 
In general, it is impractical because of its vulnerability to strategic misreporting. 
On the other hand, the second-best EDIM ensures that, irrespective of the reported types, the equilibrium eco-driving profile of drivers in the induced eco-driving game is larger than the recommended levels obtained by the mechanism.

This work offers a mathematical framework for promoting eco-driving that considers both network effects and driver behavior. 
By incorporating strategic considerations and individual preferences, the proposed second-best mechanism provides a robust and effective approach to achieving sustainable urban transportation. 
However, before real-world implementation can be considered, further research is needed regarding (1)~estimating emissions and travel time functions, (2)~handling uncertainties in data, and (3)~mapping driving policies of drivers to their eco-driving levels. 
We also remark that the proposed incentive mechanisms are direct, i.e., drivers must directly report their types/preferences to the TSO, which may not be very practical because drivers themselves may not know their types. 
In such a case, one could consider indirect or learning-based mechanisms, where the TSO infers the types of drivers by observing their driving behaviors. 
However, from the analysis point of view, the implementation of both mechanisms is equivalent thanks to the revelation principle.
Nonetheless, indirect mechanisms are more practical, which will be explored in our future work.

\appendix

\subsection{Proof of \cref{prop:IC-char}}
\label{appendix:prop:IC}

By adding and subtracting $c_i(f(\hat\theta_i,\theta_{-i}),\hat\theta_i)$ on the right side of \eqref{eq:IC-def}, we can write the condition of truthfulness as
\begin{multline}
    \label{eq:IC-def2}
    \ell_i(f(\theta_i,\theta_{-i}),\theta_i,u_i) \leq \ell_i(f(\hat\theta_i,\theta_{-i}),\hat\theta_i,u_i) \\
    + (\theta_i - \hat\theta_i) (x_i(f(\hat\theta_i,\theta_{-i})) - y_i(f(\hat\theta_i,\theta_{-i}))).
\end{multline}
If $\ell_i(f(\theta_i,\theta_{-i}),\theta_i,u_i)$ is concave in $\theta_i$ and if \eqref{eq:ic-char} holds, then \eqref{eq:IC-def2} is satisfied because of the concavity characterization of differentiable functions.

Now, suppose $\ell_i(f(\theta_i, \theta_{-i}), \theta_i,u_i)$ is not concave in $\theta_i$ but suppose \eqref{eq:ic-char} holds, then there exists $\theta_i\in\Theta_i$ such that the truthfulness condition \eqref{eq:IC-def2} is violated because of non-concavity of $\ell_i$. Finally, suppose $\ell_i(f(\theta_i,\theta_{-i}), \theta_i,u_i)$ is concave in $\theta_i$ but suppose \eqref{eq:ic-char} does not hold. Then, for some $\theta_i\in\Theta_i$, either of the following holds:
\begin{subequations}
    \begin{align}
        & \D_{\theta_i} \ell_i(f(\theta),\theta_i,u_i) < x_i(f(\theta)) - y_i(f(\theta))
        \tag{C1} \label{eq:cond1} \\ 
        & \D_{\theta_i} \ell_i(f(\theta),\theta_i,u_i) > x_i(f(\theta)) - y_i(f(\theta)).
        \tag{C2} \label{eq:cond2}
    \end{align}
\end{subequations}
For the sake of showing a contradiction, assume that the truthfulness condition \eqref{eq:IC-def2} holds. Then, add and subtract $(\theta_i-\hat\theta_i) \D_{\theta_i} \ell_i(f(\theta),\theta_i,u_i)$ in \eqref{eq:IC-def2} and rearrange. By the concavity of $\ell_i$ with respect to $\theta_i$, we have the right-hand side of \eqref{eq:IC-def2}
\begin{multline*}
    \ell_i(f(\hat\theta_i,\theta_{-i}),\hat\theta_i,u_i) - \ell_i(f(\theta_i,\theta_{-i}),\theta_i,u_i)\\ 
    + (\theta_i-\hat\theta_i) \D_{\theta_i} \ell_i(f(\theta),\theta_i,u_i) \leq 0
\end{multline*}
but for either ``condition \eqref{eq:cond1} and $\hat\theta_i > \theta_i$'' or ``condition \eqref{eq:cond2} and $\hat\theta_i < \theta_i$'', 
we have the left-hand side of \eqref{eq:IC-def2}
\begin{multline*}
    (\theta_i - \hat\theta_i) (\D_{\theta_i} \ell_i(f(\theta),\theta_i,u_i) 
    - x_i(f(\hat\theta_i,\theta_{-i})) \\ + y_i(f(\hat\theta_i,\theta_{-i}))) > 0.
\end{multline*}
This is a contradiction because a positive number can never be less than equal to a non-positive number. 
\qed


\begin{thebibliography}{10}

\bibitem{epa-ghg2025}
``Sources of greenhouse gas emissions.'' {\it United States Environmental
  Protection Agency}.
\newblock
  \href{https://www.epa.gov/ghgemissions/sources-greenhouse-gas-emissions}{Available
  Online (Accessed April 2, 2025)}.

\bibitem{ipcc-2022}
{Intergovernmental Panel on Climate Change (IPCC), Ed.}, ``Transport,'' in {\em
  Climate Change 2022 - Mitigation of Climate Change: Working Group III
  Contribution to the IPCC Sixth Assessment Report}, pp.~1049--1160, Cambridge:
  Cambridge University Press, 2023.

\bibitem{liotta2023}
C.~Liotta, V.~Vigui{\'e}, and F.~Creutzig, ``Environmental and welfare gains
  via urban transport policy portfolios across 120 cities,'' {\em Nature
  Sustainability}, vol.~6, no.~9, pp.~1067--1076, 2023.

\bibitem{huang2018}
Y.~Huang, E.~C. Ng, J.~L. Zhou, N.~C. Surawski, E.~F. Chan, and G.~Hong,
  ``Eco-driving technology for sustainable road transport: A review,'' {\em
  Renewable and Sustainable Energy Reviews}, vol.~93, pp.~596--609, 2018.

\bibitem{jayawardana2022}
V.~Jayawardana and C.~Wu, ``Learning eco-driving strategies at signalized
  intersections,'' in {\em European Control Conference (ECC)}, pp.~383--390,
  2022.

\bibitem{sivak2012}
M.~Sivak and B.~Schoettle, ``Eco-driving: Strategic, tactical, and operational
  decisions of the driver that influence vehicle fuel economy,'' {\em Transport
  Policy}, vol.~22, pp.~96--99, 2012.

\bibitem{zhang2015}
R.~Zhang and E.~Yao, ``Electric vehicles' energy consumption estimation with
  real driving condition data,'' {\em Transportation Research Part D: Transport
  and Environment}, vol.~41, pp.~177--187, 2015.

\bibitem{holdway2010}
A.~R. Holdway, A.~R. Williams, O.~R. Inderwildi, and D.~A. King, ``Indirect
  emissions from electric vehicles: Emissions from electricity generation,''
  {\em Energy \& Environmental Science}, vol.~3, no.~12, pp.~1825--1832, 2010.

\bibitem{niazi2024}
M.~U.~B. Niazi, J.-H. Cho, M.~A. Dahleh, R.~Dong, and C.~Wu, ``Incentive design
  for eco-driving in urban transportation networks,'' in {\em European Control
  Conference (ECC)}, pp.~986--991, 2024.

\bibitem{cho2024}
J.-H. Cho, M.~U.~B. Niazi, S.~Du, T.~Zhou, R.~Dong, and C.~Wu, ``Learning-based
  incentive design for eco-driving guidance,'' in {\em Conference in Emerging
  Technologies in Transportation Systems (TRC-30)}, 2024.

\bibitem{lai2015}
W.-T. Lai, ``The effects of eco-driving motivation, knowledge and reward
  intervention on fuel efficiency,'' {\em Transportation Research Part D:
  Transport and Environment}, vol.~34, pp.~155--160, 2015.

\bibitem{liimatainen2011}
H.~Liimatainen, ``Utilization of fuel consumption data in an ecodriving
  incentive system for heavy-duty vehicle drivers,'' {\em IEEE Transactions on
  Intelligent Transportation Systems}, vol.~12, no.~4, pp.~1087--1095, 2011.

\bibitem{schall2017}
D.~L. Schall and A.~Mohnen, ``Incentivizing energy-efficient behavior at work:
  An empirical investigation using a natural field experiment on eco-driving,''
  {\em Applied Energy}, vol.~185, pp.~1757--1768, 2017.

\bibitem{mcconky2018}
K.~McConky, R.~B. Chen, and G.~R. Gavi, ``A comparison of motivational and
  informational contexts for improving eco-driving performance,'' {\em
  Transportation Research Part F: Traffic Psychology and Behaviour}, vol.~52,
  pp.~62--74, 2018.

\bibitem{vaezipour2019}
A.~Vaezipour, A.~Rakotonirainy, N.~Haworth, and P.~Delhomme, ``A simulator
  study of the effect of incentive on adoption and effectiveness of an
  in-vehicle human machine interface,'' {\em Transportation Research Part F:
  Traffic Psychology and Behaviour}, vol.~60, pp.~383--398, 2019.

\bibitem{schwarting2019}
W.~Schwarting, A.~Pierson, J.~Alonso-Mora, S.~Karaman, and D.~Rus, ``Social
  behavior for autonomous vehicles,'' {\em Proceedings of the National Academy
  of Sciences}, vol.~116, no.~50, pp.~24972--24978, 2019.

\bibitem{ozkan2021}
M.~F. Ozkan and Y.~Ma, ``Socially compatible control design of automated
  vehicle in mixed traffic,'' {\em IEEE Control Systems Letters}, vol.~6,
  pp.~1730--1735, 2021.

\bibitem{jayawardana2023}
S.~Jayawardana, V.~Jayawardana, K.~Vidanage, and C.~Wu, ``Multi-behavior
  learning for socially compatible autonomous driving,'' in {\em 26th IEEE
  International Conference on Intelligent Transportation Systems (ITSC)},
  pp.~4422--4427, 2023.

\bibitem{wu2019}
M.~Wu and S.~Amin, ``Information design for regulating traffic flows under
  uncertain network state,'' in {\em Annual Allerton Conference on
  Communication, Control, and Computing (Allerton)}, pp.~671--678, 2019.

\bibitem{zhu2022}
Y.~Zhu and K.~Savla, ``Information design in nonatomic routing games with
  partial participation: Computation and properties,'' {\em IEEE Transactions
  on Control of Network Systems}, vol.~9, no.~2, pp.~613--624, 2022.

\bibitem{massicot2021}
O.~Massicot and C.~Langbort, ``Competitive comparisons of strategic information
  provision policies in network routing games,'' {\em IEEE Transactions on
  Control of Network Systems}, vol.~9, no.~4, pp.~1589--1599, 2021.

\bibitem{ferguson2024}
B.~L. Ferguson, P.~N. Brown, and J.~R. Marden, ``Information signaling with
  concurrent monetary incentives in {Bayesian} congestion games,'' {\em IEEE
  Transactions on Intelligent Transportation Systems}, 2024.

\bibitem{bonatti2023}
A.~Bonatti, M.~Dahleh, T.~Horel, and A.~Nouripour, ``Coordination via selling
  information,'' {\em arXiv:2302.12223}, 2023.

\bibitem{bonatti2024}
A.~Bonatti, M.~Dahleh, T.~Horel, and A.~Nouripour, ``Selling information in
  competitive environments,'' {\em Journal of Economic Theory}, vol.~216,
  p.~105779, 2024.

\bibitem{satchidanandan2022-2}
B.~Satchidanandan, M.~Roozbehani, and M.~A. Dahleh, ``A two-stage mechanism for
  demand response markets,'' {\em IEEE Control Systems Letters}, vol.~7,
  pp.~49--54, 2022.

\bibitem{satchidanandan2023}
B.~Satchidanandan and M.~A. Dahleh, ``Incentive compatibility in two-stage
  repeated stochastic games,'' {\em IEEE Transactions on Control of Network
  Systems}, vol.~11, no.~1, pp.~295--306, 2023.

\bibitem{dahleh2024}
M.~A. Dahleh, T.~Horel, and M.~U.~B. Niazi, ``Mitigating information asymmetry
  in two-stage contracts with non-myopic agents,'' {\em IFAC-PapersOnLine},
  vol.~58, no.~30, pp.~19--24, 2024.

\bibitem{satchidanandan2022}
B.~Satchidanandan and M.~A. Dahleh, ``An efficient and incentive-compatible
  mechanism for energy storage markets,'' {\em IEEE Transactions on Smart
  Grid}, vol.~13, no.~3, pp.~2245--2258, 2022.

\bibitem{kamal2010}
M.~A.~S. Kamal, M.~Mukai, J.~Murata, and T.~Kawabe, ``On board eco-driving
  system for varying road-traffic environments using model predictive
  control,'' in {\em IEEE International Conference on Control Applications},
  pp.~1636--1641, 2010.

\bibitem{padilla2018}
G.~Padilla, S.~Weiland, and M.~Donkers, ``A global optimal solution to the
  eco-driving problem,'' {\em IEEE Control Systems Letters}, vol.~2, no.~4,
  pp.~599--604, 2018.

\bibitem{han2019}
J.~Han, A.~Vahidi, and A.~Sciarretta, ``Fundamentals of energy efficient
  driving for combustion engine and electric vehicles: An optimal control
  perspective,'' {\em Automatica}, vol.~103, pp.~558--572, 2019.

\bibitem{mintsis2020}
E.~Mintsis, E.~I. Vlahogianni, and E.~Mitsakis, ``Dynamic eco-driving near
  signalized intersections: Systematic review and future research directions,''
  {\em Journal of Transportation Engineering, Part A: Systems}, vol.~146,
  no.~4, p.~04020018, 2020.

\bibitem{sun2020}
C.~Sun, J.~Guanetti, F.~Borrelli, and S.~J. Moura, ``Optimal eco-driving
  control of connected and autonomous vehicles through signalized
  intersections,'' {\em IEEE Internet of Things Journal}, vol.~7, no.~5,
  pp.~3759--3773, 2020.

\bibitem{turri2016}
V.~Turri, B.~Besselink, and K.~H. Johansson, ``Cooperative look-ahead control
  for fuel-efficient and safe heavy-duty vehicle platooning,'' {\em IEEE
  Transactions on Control Systems Technology}, vol.~25, no.~1, pp.~12--28,
  2016.

\bibitem{homchaudhuri2016}
B.~HomChaudhuri, A.~Vahidi, and P.~Pisu, ``Fast model predictive control-based
  fuel efficient control strategy for a group of connected vehicles in urban
  road conditions,'' {\em IEEE Transactions on Control Systems Technology},
  vol.~25, no.~2, pp.~760--767, 2016.

\bibitem{de-nunzio2016}
G.~De~Nunzio, C.~C. De~Wit, P.~Moulin, and D.~Di~Domenico, ``Eco-driving in
  urban traffic networks using traffic signals information,'' {\em
  International Journal of Robust and Nonlinear Control}, vol.~26, no.~6,
  pp.~1307--1324, 2016.

\bibitem{scora2006}
G.~Scora and M.~Barth, ``Comprehensive modal emissions model {(CMEM)}, version
  3.01,'' {\em User Guide: Centre for Environmental Research and Technology.
  University of California, Riverside}, 2006.

\bibitem{park2013}
S.~Park, H.~Rakha, K.~Ahn, and K.~Moran, ``{Virginia Tech} comprehensive
  power-based fuel consumption model {(VT-CPFM)}: Model validation and
  calibration considerations,'' {\em International Journal of Transportation
  Science and Technology}, vol.~2, no.~4, pp.~317--336, 2013.

\bibitem{singh2021}
M.~Singh and R.~K. Dubey, ``Deep learning model based {CO2} emissions
  prediction using vehicle telematics sensors data,'' {\em IEEE Transactions on
  Intelligent Vehicles}, vol.~8, no.~1, pp.~768--777, 2021.

\bibitem{allstrom2012}
A.~Allstr{\"o}m, D.~Gundleg{\aa}rd, and C.~Rydergren, ``Evaluation of travel
  time estimation based on {LWR-v} and {CTM-v}: A case study in {Stockholm},''
  in {\em 15th International IEEE Conference on Intelligent Transportation
  Systems}, pp.~1644--1649, 2012.

\bibitem{li2017}
Y.~Li, D.~Gunopulos, C.~Lu, and L.~Guibas, ``Urban travel time prediction using
  a small number of {GPS} floating cars,'' in {\em 25th ACM SIGSPATIAL
  International Conference on Advances in Geographic Information Systems},
  pp.~1--10, 2017.

\bibitem{wahlstrom2017}
J.~Wahlstr{\"o}m, I.~Skog, and P.~H{\"a}ndel, ``Smartphone-based vehicle
  telematics: A ten-year anniversary,'' {\em IEEE Transactions on Intelligent
  Transportation Systems}, vol.~18, no.~10, pp.~2802--2825, 2017.

\bibitem{gao2022}
G.~Gao, S.~Meng, and M.~V. W{\"u}thrich, ``What can we learn from telematics
  car driving data: A survey,'' {\em Insurance: Mathematics and Economics},
  vol.~104, pp.~185--199, 2022.

\bibitem{yin2006}
Y.~Yin and S.~Lawphongpanich, ``Internalizing emission externality on road
  networks,'' {\em Transportation Research Part D: Transport and Environment},
  vol.~11, no.~4, pp.~292--301, 2006.

\bibitem{boriboonsomsin2010}
K.~Boriboonsomsin, A.~Vu, and M.~Barth, ``Eco-driving: Pilot evaluation of
  driving behavior changes among {U.S.} drivers,'' {\em UC Berkeley: University
  of California Transportation Center}, 2010.

\bibitem{coloma2017}
J.~F. Coloma, M.~Garc{\'\i}a, Y.~Wang, and A.~Monz{\'o}n, ``Green eco-driving
  effects in non-congested cities,'' {\em Sustainability}, vol.~10, no.~1,
  p.~28, 2017.

\bibitem{debreu1952}
G.~Debreu, ``A social equilibrium existence theorem,'' {\em Proceedings of the
  National Academy of Sciences}, vol.~38, no.~10, pp.~886--893, 1952.

\bibitem{rosen1965}
J.~B. Rosen, ``Existence and uniqueness of equilibrium points for concave
  n-person games,'' {\em Econometrica: Journal of the Econometric Society},
  pp.~520--534, 1965.

\bibitem{laffont2002}
J.-J. Laffont and D.~Martimort, {\em The theory of incentives: The
  principal-agent model}.
\newblock Princeton University Press, 2002.

\bibitem{ratliff2013}
L.~J. Ratliff, S.~A. Burden, and S.~S. Sastry, ``Characterization and
  computation of local {Nash} equilibria in continuous games,'' in {\em Annual
  Allerton Conference on Communication, Control, and Computing (Allerton)},
  pp.~917--924, 2013.

\bibitem{mertikopoulos2017}
P.~Mertikopoulos and M.~Staudigl, ``Convergence to {Nash} equilibrium in
  continuous games with noisy first-order feedback,'' in {\em IEEE Conference
  on Decision and Control (CDC)}, pp.~5609--5614, 2017.

\bibitem{mertikopoulos2019}
P.~Mertikopoulos and Z.~Zhou, ``Learning in games with continuous action sets
  and unknown payoff functions,'' {\em Mathematical Programming}, vol.~173,
  pp.~465--507, 2019.

\bibitem{toonsi2024}
S.~Toonsi and J.~Shamma, ``Higher-order uncoupled dynamics do not lead to
  {Nash} equilibrium--except when they do,'' {\em Advances in Neural
  Information Processing Systems}, vol.~36, 2024.

\bibitem{liu2016}
J.~Liu, X.~Wang, and A.~Khattak, ``Customizing driving cycles to support
  vehicle purchase and use decisions: Fuel economy estimation for alternative
  fuel vehicle users,'' {\em Transportation Research Part C: Emerging
  Technologies}, vol.~67, pp.~280--298, 2016.

\bibitem{kawamoto2019}
R.~Kawamoto, H.~Mochizuki, Y.~Moriguchi, T.~Nakano, M.~Motohashi, Y.~Sakai, and
  A.~Inaba, ``Estimation of {CO2} emissions of internal combustion engine
  vehicle and battery electric vehicle using {LCA},'' {\em Sustainability},
  vol.~11, no.~9, p.~2690, 2019.

\bibitem{barth2009}
M.~Barth and K.~Boriboonsomsin, ``Energy and emissions impacts of a
  freeway-based dynamic eco-driving system,'' {\em Transportation Research Part
  D: Transport and Environment}, vol.~14, no.~6, pp.~400--410, 2009.

\bibitem{yang2016}
H.~Yang, H.~Rakha, and M.~V. Ala, ``Eco-cooperative adaptive cruise control at
  signalized intersections considering queue effects,'' {\em IEEE Transactions
  on Intelligent Transportation Systems}, vol.~18, no.~6, pp.~1575--1585, 2016.

\end{thebibliography}
\end{document}